\documentclass[10pt,journal,twocolumn]{IEEEtran}
\usepackage{amsmath}%
\usepackage{amsfonts}%
\usepackage{amssymb}%
\usepackage{graphicx}%
\usepackage{array}%
\usepackage{subfig}%
\usepackage{array}%

\newtheorem{theorem}{Theorem}

\newtheorem{assumptions}{Assumption}

\newtheorem{definition}{Definition}

\newtheorem{lemma}{Lemma}

\newenvironment{proof}[1][Proof]{\textbf{#1.} }{\ \rule{0.5em}{0.5em}}

\newcommand{\Tmin}{{T_*}}
\newcommand{\Tmax}{{T^*}}
\newcommand{\out}{\mathcal{O}}
\newcommand{\outm}{\overline{\mathcal{O}}_m}
\newcommand{\gs}{\widetilde{\gamma}}
\newcommand{\gl}{\overline{\gamma}}
\newcommand{\ds}{\widetilde{d}}
\newcommand{\dl}{\overline{d}}
\newcommand{\Ql}{\overline{Q}}
\newcommand{\Qs}{\widetilde{Q}}
\newcommand{\Pl}{P(\overline{\out})}
\newcommand{\PlN}{P_N(\overline{\out})}
\newcommand{\PlPd}{P_{PD}(\overline{\out})}
\newcommand{\PlPr}{P_{PR}(\overline{\out})}
\newcommand{\PlNs}{P_N(\overline{\out}^s)}

\newcommand{\yl}{\overline{y}}

\newcommand{\Ps}{P(\widetilde{\out})}
\newcommand{\PsN}{P_N(\widetilde{\out})}
\newcommand{\PsNs}{P_N(\widetilde{\out}^s)}

\newcommand{\PsPd}{P_{PD}(\widetilde{\out})}
\newcommand{\PsPr}{P_{PR}(\widetilde{\out})}
\newcommand{\Xl}{\overline{X}}
\newcommand{\Sl}{\overline{S}}

\newcommand{\thl}{\overline{\theta}}

\newcommand{\I}{\mathcal{I}}
\newcommand{\J}{\mathcal{J}}
\newcommand{\Ur}{\mathcal{U}_R}
\newcommand{\Lr}{\mathcal{L}_R}
\newcommand{\vl}{\overline{v}}
\newcommand{\vs}{\widetilde{v}}
\newcommand{\ml}{\overline{m}}
\newcommand{\ms}{\widetilde{m}}
\newcommand{\al}{\overline{\alpha}}
\newcommand{\as}{\widetilde{\alpha}}

\newcommand{\vp}{\mathbf{p}}

\newcommand{\gm}{\overline{\gamma^m}}
\newcommand{\gu}{\overline{\gamma^u}}
\newcommand{\xm}{\xi^m_T}
\newcommand{\outa}{\overline{\mathcal{O}}_{A}}

\begin{document}

\title{Proactive Resource Allocation: Harnessing the Diversity and Multicast Gains}

\author{
\IEEEauthorblockN{John Tadrous, Atilla Eryilmaz, and Hesham El Gamal}
\thanks{Authors are with the Department of Electrical and Computer Engineering at the Ohio State University, Columbus, USA.\newline \indent E-mail: \{tadrousj,eryilmaz,helgamal\}@ece.osu.edu}}
\maketitle

\begin{abstract}
This paper introduces the novel concept of proactive resource allocation through which the predictability of user behavior is exploited to balance the wireless traffic over time, and hence, significantly reduce the bandwidth required to achieve a given blocking/outage probability. We start with a simple model in which the smart wireless devices are assumed to predict the arrival of new requests and submit them to the network $T$ time slots in advance. Using tools from large deviation theory, we quantify the resulting {\bf prediction diversity gain} to establish that the decay rate of the outage event probabilities increases with the prediction duration $T$. This model is then generalized to incorporate the effect of the randomness in the prediction look-ahead time $T$. Remarkably, we also show that, in the cognitive networking scenario, the appropriate use of proactive resource allocation by the primary users improves the diversity gain of the secondary network at no cost in the primary network diversity. We also shed lights on multicasting with predictable demands and show that the proactive multicast networks can achieve a significantly higher diversity gain that scales super-linearly with $T$. Finally, we conclude by a discussion of the new research questions posed under the umbrella of the proposed {\bf proactive (non-causal) wireless networking} framework.
\end{abstract}

\begin{IEEEkeywords}
Scheduling, large deviations, diversity gain, multicast alignment, predictive traffic. 
\end{IEEEkeywords}

\section{Introduction}
\IEEEPARstart{I}deally, wireless networks should be optimized to deliver the best Quality of Service (in terms of reliability, delay, and throughput) to the subscribers with the minimum expenditure in resources. Such resources include transmitted power, transmitter and receiver complexity, and allocated frequency spectrum. Over the last few years, we have experienced an ever increasing demand for wireless spectrum resulting from the adoption of {\em throughput hungry} applications in a variety of civilian, military, and scientific settings.

Since the available spectrum is non renewable and limited, this demand motivates the need for efficient wireless networks that {\em maximally utilize} the spectrum. In this work, we focus our attention on the resource allocation aspect of the problem and propose a new paradigm that offers remarkable spectral gains in a variety of relevant scenarios. More specifically, our proactive resource allocation framework exploits the predictability of our daily usage of wireless devices to smooth out the traffic demand in the network, and hence, reduce the required resources to achieve a certain point on the Quality of Service (QoS) curve. This new approach is motivated by the following observations.

\newcounter{ctr0}
\setcounter{ctr0}{1}
\noindent $\bullet$ While there is a severe shortage in the spectrum, it is well-documented that a significant fraction of the available spectrum
is under-utilized \cite{FCC2002}. In fact, this is the main
motivation for the cognitive networking where secondary
users are allowed to use the spectrum at the off time of the primary so as to maximize the spectral
efficiency~\cite{Mitola}. The cognitive radio
approach, however, is still facing significant technological
hurdles~\cite{Ian},~\cite{Gridlock} and, will offer only a
partial solution to the problem. This limitation is tied to the main reason behind the
under-utilization of the spectrum; namely \emph{the large disparity
between the average and peak traffic demand in the network}.

 Actually, one can see
that the traffic demand at the peak hours is much higher than that
at night. Now, the cognitive radio approach assumes that
the secondary users will be able to utilize the spectrum at the off-peak times, but at those times one may
expect the secondary traffic characteristics to be similar to that
of the primary users (e.g., at night most of the primary and
secondary users are expected to be idle). Thereby, the overarching goal of the proactive resource allocation framework is to avoid this limitation, and hence, achieve a significant reduction in the peak to average demand ratio {\em without relying on out of network users}.

\addtocounter{ctr0}{1}\noindent
$\bullet$ In the traditional approach, wireless networks are
constructed assuming that the subscribers are equipped with {\em
dumb terminals} with very limited computational power. It is obvious
that the new generation of {\em smart devices} enjoy significantly
enhanced capabilities in terms of both { processing power and
available memory}.This observation should inspire a similar paradigm shift in the
design of wireless networks whereby the capabilities of the smart
wireless terminals are leveraged to maximize the utility of the
frequency spectrum, {\em a non-renewable resource that does not
scale according to Moore's law}. Our proactive resource allocation framework is a significant step in this direction.


\addtocounter{ctr0}{1} \noindent
$\bullet$ The introduction of smart phones has resulted in a paradigm shift in the dominant traffic in mobile
cellular networks. While the primary traffic source in traditional
cellular networks was { real-time} voice communication, one can
argue that a significant fraction of the traffic generated by the
smart phones results from non-real-time data requests (e.g., file
downloads). As demonstrated in the following, this feature allows for more degrees of freedom in the design of the scheduling algorithm.

\addtocounter{ctr0}{1}\noindent
$\bullet$ The final piece of our puzzle relates to the observation
that the human usage of the wireless devices is {\em highly predictable}.
This claim is supported by a growing body of evidence that ranges
from the recent launch of {\bf Google Instant} to the interesting
findings on our predictable mobility patterns~\cite{SQBB10}. An example would be the fact that our
preference for a particular news outlet is not expected to change
frequently. So, if the smart phone observes that the user is
downloading CNN, for example, in the morning for a sequence of days
in a row then it can {\bf anticipate} that the user will be
interested in the CNN again the following day. One
 can now see the potential for scheduling early downloads of the
predictable traffic to {\em reduce the peak to average traffic
demand} by maximally exploiting the available spectrum in the network
idle time.

These observations motivate us in this work to develop and analyze proactive resource allocation strategies in the presence of user predictability under various conditions, dynamics, and operational capabilities. In particular, our contributions along with their position in the rest of the paper are:

  \noindent $\bullet$ In Section \ref{sec:sys_model} we state the predictive network model and introduce the outage probability and the associated diversity gain for two main scaling regimes, namely, linear and polynomial scaling.

  \noindent $\bullet$ In Section \ref{sec:Div_gain}, we establish the diversity gain of non-predictive and predictive networks, and analyze the effect of the random look-ahead window size, $T$. Our analysis reveals a minimum improvement factor of (1+T) in the diversity gain for both linear and polynomial scaling regimes.

 \noindent $\bullet$ In Section \ref{sec:Diff_QoS_Users}, we investigate proactive scheduling in a two-QoS network,typical of a cognitive radio network, where we prove the existence of a proactive scheduling policy that can maintain the diversity gain level of the primary predictive network while strictly improving it for the secondary non-predictive network.

  \noindent $\bullet$ In Section \ref{sec:Robustness}, we analyze the robustness of the proactive resource allocation scheme to the prediction errors and determine the optimal choice of the look-ahead window size given an imperfect prediction mechanism to maximize the diversity gain, which is shown to be always strictly greater than that of the non-predictive network.

  \noindent $\bullet$ In Section \ref{sec:Multicast}, we analyze the proactive multicasting with predictable demands, and show the significant gains that can be leveraged through the alignment property offered by predictable multicast traffic. More specifically, we show that the diversity gain of a proactive multicasting network is increasing super-linearly with the window size, $T$, for the linear scaling regime.

  \noindent $\bullet$ In Section \ref{sec:Conc}, we conclude the paper and highlight other important research aspects that can be leveraged through predictive wireless communications.


The proactive wireless network can be viewed as an ordinary network with delay tolerant requests, that is, when the network predicts a request a head of time, the actual arrival time  of that request can be considered as a hard deadline that the scheduler should meet. In \cite{Jindal}, scheduling with deadlines was considered for a single packet under the objective of minimizing the expected energy consumed for transmission. In \cite{Javidi}, the asymptotic performance of the error probability with the signal-to-noise ratio was analyzed when the bits of each codeword must be delivered under hard deadline constraints. In \cite{s_deadlines} and \cite{s_scheduling}, scheduling with deadlines was also addressed from queuing theory point of view under different objectives and multiple priority classes while optimal scheduling policies were investigated for different scenarios.

In this paper, however, we look at the scheduling problem with deadlines from a different perspective, where we define the outage probability as the probability of having a time slot suffering expiring requests, and we analyze the asymptotic decay rate of this outage probability with the system capacity, $C$, when the input traffic is increasing in $C$ either linearly or polynomially, and $C$ is approaching infinity. We call this metric the diversity gain of the network and show that its behavior can significantly be improved by exploiting the predictable behavior of the users. This metric and line of investigation are also motivated by the order-wise difference between the timescale of the prediction window lengths (typically  of the order of tens of minutes, if not hours) and the timescale of application-based deadline-constraints (of the order of milliseconds) considered in other works.     

\section{System Model}
\label{sec:sys_model}
We consider a simplified model of a single server, time-slotted wireless network where the requests arrive at the beginning of each slot. The number of arriving requests at time slot $n$ is an integer-valued random variable denoted by $Q(n)$ that is assumed to be ergodic and Poisson distributed with mean $\lambda$. Each request is assumed to consume one unit of resource and is completely served in a single time slot. Moreover, the wireless network has a \emph{fixed} capacity $C$ per slot. We distinguish two types of wireless resource allocation: \textbf{reactive} and \textbf{proactive}. In reactive  resource allocation, the wireless network responds to user requests right after they are initiated by the user, whereas in proactive resource allocation, the network can track, learn and then \emph{predict} the user requests ahead of time, and hence possesses more flexibility in scheduling these requests before their actual time of arrival. We refer to the networks that perform reactive and proactive resource allocation, respectively, as \textbf{non-predictive} and \textbf{predictive} networks.
 
The predictive wireless network can anticipate the arrival of requests a number of time slots ahead. That is, if $q(n)$, $q\in\{1, \cdots, Q(n)\}$, is the identifier of a request predicted at the beginning of time slot $n$, the predictive network has the advantage of serving this request within the next $T_{q(n)}$ slots. Hence, when request $q(n)$ arrives at a predictive network, it has a deadline at time slot $D_{q(n)}=n+T_{q(n)}$ as shown in Fig. \ref{fig:Model_1}. 
\begin{figure}
	\centering
		\includegraphics[width=0.40\textwidth]{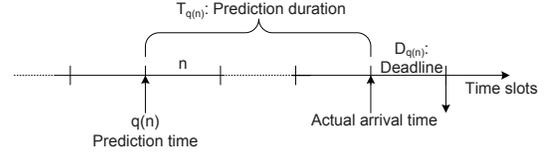}
	\caption{Prediction Model: $q(n)$ is a request predicted at the beginning of time slot $n$, $T_{q(n)}$ is the prediction duration of request $q(n)$, and $D_{q(n)}$ is the actual slot of arrival for request $q(n)$ which can be considered as the deadline for such a request.}
	\label{fig:Model_1}
\end{figure}

Conversely, in a non-predictive network, all arriving requests at the beginning of time slot $n$ must be served in the same time slot $n$, i.e., if $q(n)$ is an unpredicted request, then $T_{q(n)}=0$ and $D_{q(n)}=n$. At this point, we wish to stress the fact that the model operates as the time scale of the application layer at which 1) the current paradigm, i.e., non-predictive networking, treats all the requests as urgent, 2) each slot duration may be in the order of minutes and possibly hours, and 3) the system capacity is fixed since the channel fluctuation dynamics are averaged out at this time scale.

\begin{definition}
\label{def:out_event}
Let $N_0(n)$ be the number of requests in the system at the beginning of time slot $n$ having a deadline of $n$. The outage event $\mathcal{O}$ is then defined as
\begin{equation}
\label{eq:out_event}
\mathcal{O}\triangleq \left\{N_0(n)>C, n\gg 1\right\},
\end{equation}
\end{definition}
The above definition states that an outage occurs at slot $n$ if and only if at least one of the requests in the system expires in this slot. The term $N_0(n)$ coincides on $Q(n)$ when the network is non-predictive, and is different when the network is predictive.

Following the definition of the outage event, we denote the probability that the wireless network runs into an outage at slot $n>0$ by $P(\out)$. Throughout this paper, we will focus on analyzing the asymptotic decay rate of the outage probability with the system capacity $C$ when it approaches infinity. We call this decay rate the {\bf diversity gain} of the network. Moreover, in our analysis we assume that the mean input traffic $\lambda$ scales with the system capacity in two different regimes as follows.
\begin{enumerate}
\item{\emph{Linear Scaling:}}
In this regime, the arrival process $\Ql(n), n>0$ is Poisson with rate $\bar{\lambda}$ that scales with $C$ as $$\bar{\lambda}=\gl C, \quad 0<\gl<1.$$ And with outage probability denoted by $\Pl$, the associated diversity gain is defined as $$\dl(\gl)\triangleq \lim_{C\rightarrow\infty}-\frac{\log \Pl}{C}.$$ 
\item{\emph{Polynomial Scaling:}}
In this regime, the arrival process $\Qs(n), n>0$, is also Poisson with rate $\tilde{\lambda}$, but the rate scales with the system capacity $C$ polynomially as $$\tilde{\lambda}=C^{\gs},\quad 0<\gs<1.$$
And with outage probability $\Ps$, the associated diversity gain is defined as
$$\ds(\gs)\triangleq \lim_{C\rightarrow\infty}-\frac{\log \Ps}{C\log C}.$$
\end{enumerate}
We consider the linear scaling of the input traffic with the system resources because it is commonly used in networking literature where the parameter $\gl$ serves as bandwidth utilization factor. As $\gl$ approaches $1$ the average input traffic approaches the capacity and the system becomes critically stable and more subject to outage events, whereas small values of $\gl$ imply underutilized resources but small probability of outage. The polynomial scaling regime is also introduced because under this type of scaling, the optimal prediction diversity gain can be fully determined through the asymptotic analysis of simple scheduling policies like earliest deadline first. Except for Section \ref{sec:Multicast} and its associated appendices, we consistently use the accents $\bar{.}$ and $\tilde{.}$ to denote linear and polynomial scaling regimes respectively, while symbols without accents are used to denote a general case.

\section{Diversity Gain Analysis}
\label{sec:Div_gain}

\subsection{Diversity Gain of Reactive Networks}
\label{subsec:Reactive_Networks}
The reactive networks are supposed to have no prediction capabilities so they cannot serve any request prior to its time of actual arrival. Hence, the reactive network encounters an outage at time slot $n$ if and only if $Q(n)>C$ as $N_0(n)=Q(n)$.

\begin{theorem}
\label{th:diversity_reac}
Denote the outage probability and the diversity gain of the non-predictive network, respectively, by $P_N(\out)$ and $d_N(\gamma)$, then 
\begin{equation}
\label{eq:div_rea_lin}
\dl_N(\gl)=\gl-1-\log \gl, \quad  0<\gl<1,
\end{equation}
and
\begin{equation}
\label{eq:div_rea_sub}
\ds_N(\gs)=1-\gs, \quad 0<\gs<1.
\end{equation}
\end{theorem}

\begin{proof}
Please refer to Appendix \ref{app:non_pred}.
\end{proof}

It can be noted that as $\gl$ and $\gs$ approach $1$, the corresponding diversity gains $\dl_N(\gl)$ and $\ds_N(\gs)$ approach $0$, as in this case the arrival rate in both regimes matches the system capacity, and hence the system becomes critically stable and the logarithm of the outage probability does not decay with $C$. However, the behavior of the the diversity gain is not the same when both $\gl$ and $\gs$ approach $0$. As $\gl\rightarrow 0$, $\dl_N(\gl)\rightarrow\infty$ because the arrival rate $\bar{\lambda}\rightarrow 0$, thus the resulting outage probability approaches $0$ and the diversity gain approaches $\infty$. Whereas $\gs\rightarrow 0$ implies that $\ds_N(\gs)\rightarrow 1$ which is the case when the input traffic is still positive but does not scale with the system capacity.

\subsection{Diversity Gain of Proactive Networks}
\label{subsec:Pred_Div_Gain}
Unlike reactive networks, the proactive network has the flexibility to schedule the predicted requests in a window of time slots through some scheduling policy. Depending on the scheduling policy employed, the resulting outage probability (and of course the associated diversity gain) varies. By the term {\bf optimal} prediction diversity gain, we mean the maximum diversity gain that can be achieved by the predictive network, which corresponds to the minimum outage probability denoted $P^*_P(\out)$. 

 In our analysis, we consider, for simplicity, the earliest deadline first (EDF) scheduling policy, which has also been called in \cite{EDF} shortest time to extinction (STE). This policy, as proved in \cite{EDF}, maximizes the number of served requests under a per-request deadline constraint. Further studies on this policies can be found in \cite{s_deadlines} and \cite{EDF2}. In the proposed predictive network, the EDF scheduling policy is defined as follows.
\begin{definition}[Earlies Deadline First (EDF)]
\label{def:EDF}
Let the maximum prediction interval for a request be denoted by $T^*$, i.e., $T^*=\sup_{q,n} \left\{T_{q(n)}\right\}$, and let $N_i(n), i=0, 1,\cdots, T^* $ be the number of requests in the system at the beginning of time slot $n$ and having a deadline of $n+i$. Then, at the beginning of slot $n$, the EDF policy sorts $\left\{N_i(n)\right\}_{i=0}^{T^*}$ in an ascending order with respect to $i$, and serves them in that order until either a total of $C$ requests get served or the network completes the service of all existing requests in this slot. 
\end{definition} 

It can be noted that EDF does not necessarily minimize the outage probability as it is only concerned with maximizing the number of served requests while the outage event does not take into account the number of dropped requests. However, EDF has two main characteristics that help in analysis. Namely, it always serves requests as long as there are any, i.e., it is a work conserving policy, and it serves requests in the order of their remaining time to deadline.     

\subsubsection{Deterministic Look-ahead Time}
\label{subsubsec:Det_T}
In this scenario, $T_{q(n)}=T$ for all $q(n), n>0$ for some constant $T\geq 0$. Hence, assuming that the system employs EDF scheduling policy, we have $T^*=T$ and $N_T(n)=Q(n), n>0$. Thus, the EDF policy will reduce to first-come-first-serve (FCFS). The outage probability in this case is denoted by $P_{PD}(\out)$.
\begin{lemma}
\label{lem:Det_T}
Under EDF, let $$\mathcal{U}_D\triangleq\left\{\sum_{i=0}^{T}{Q(n-T-i)}>C(T+1),n\gg 1\right\}$$ and 
$$\mathcal{L}_D\triangleq\left\{Q(n-T)>C(T+1), n\gg 1 \right\}.$$ Then, the events $\mathcal{U}_D$ and $\mathcal{L}_D$ constitute a necessary condition and a sufficient condition on the outage event, respectively. Hence, $P(\mathcal{L}_D)\leq P_{PD}(\out)\leq P(\mathcal{U}_D)$.
\end{lemma}
In the above lemma, we assume that $n\gg 1$ as we are interested in the steady state performance. The event $\mathcal{U}_D$ occurs when the number of arriving requests in consecutive $T+1$ slots is larger than the total capacity of $T+1$ slots, whereas the event $\mathcal{L}_D$ occurs when the number of arriving requests at any slot is larger than the total capacity of $T+1$ slots.
\begin{proof}
Please refer to Appendix \ref{app:Det_T}.
\end{proof}
It is obvious from the proof that the event $\mathcal{U}_D$ is related to the outage event $\out$ through the EDF scheduling policy, whereas the event $\mathcal{L}_D$ is independent of the scheduling policy employed. 

\begin{theorem}
\label{th:diversity_proac_d} 
The optimal prediction diversity gain of a proactive network with deterministic prediction interval $T$, denoted $d_{PD}(\gamma)$, satisfies
\begin{align}
\label{eq:dl_Pd_lower}
&\dl_{PD}(\gl)\geq (1+T)(\gl-1-\log \gl), \quad 0<\gl<1,\\
\label{eq:ds_Pd}
&\ds_{PD}(\gs)=(1+T)(1-\gs), \quad 0<\gs<1.
\end{align}
\end{theorem}
The above result shows that proactive resource allocation offers a multiplicative diversity gain of at least $T+1$ for the linear scaling regime and exactly $T+1$ for the polynomial scaling regime.
\begin{proof}
Please refer to Appendix \ref{app:div_pred_d}.
\end{proof}
Note that, an upper bound on $\dl_{PD}(\gl)$ can be established using $P(\overline{\mathcal{L}}_D)\leq\PlPd$ and following the same approach of deriving the lower bound in the theorem. This upper bound will be given by
\begin{equation}
\label{eq:dl_Pd_upper}
\dl_{PD}(\gl)\leq (T+1)\left(\frac{\gl}{T+1}-1+\log\left(\frac{T+1}{\gl}\right)\right).
\end{equation}
Comparing the right hand sides of \eqref{eq:dl_Pd_lower}, and \eqref{eq:dl_Pd_upper} it can be seen that they match only when $T=0$, and in this case, they also match the non-predictive diversity gain obtained in \eqref{eq:dl_N}. Otherwise, for positive values of $T$, the two bounds differ.

\subsubsection{Random Look-ahead Time}
We consider a more general scenario where $T_{q(n)}$, $0\leq q(n)\leq Q(n)$, $n>0$ is a sequence of IID non-negative integer-valued random variables defined over a finite support $\{\Tmin,\cdots \Tmax\}$, where $0\leq\Tmin\leq\Tmax<\infty$. The random variable $T_{q(n)}$ has the following probability mass function (PMF),
\begin{equation}
\label{eq:PMFofT}
P\left(T_{q(n)}=k\right)\triangleq\left\{
\begin{array}{l l}
p_k,\quad \quad \Tmin\leq k \leq\Tmax,\\
0,\quad \quad \text{otherwise},
\end{array}
\right.
\end{equation}
where $\sum_{k=\Tmin}^{\Tmax}p_k=1$ and $p_k\geq 0, \quad \forall k$, the cumulative distribution function (CDF) of $T_{q(n)}$ can be written as
\begin{equation}
\label{eq:CDFofT}
P(T_{q(n)}\leq k)=F_k=\left\{
\begin{array}{l l}
1,& k> \Tmax,\\
\sum_{i=\Tmin}^k p_i,&\Tmin\leq k \leq\Tmax,\\
0,& k<\Tmin.
\end{array}
\right.
\end{equation} 
Hence, the overall process $Q(n)$ can be decomposed to a superposition of independent Poisson processes, i.e., $$Q(n)=\sum_{k=\Tmin}^{\Tmax}Q_{k}(n)$$ where $Q_k(n)$, $n>0$ is the process of requests predicted $k$ slots ahead, $k=\Tmin, \cdots, \Tmax$. The arrival rate of $Q_k(n)$ is $p_k\lambda$.

In this scenario, we denote the outage probability under EDF by $P_{PR}(\out)$ and the optimal diversity gain by $d_{PR}(\gamma)$. Unlike the case of deterministic look-ahead time, EDF here does not reduce to FCFS because the arriving requests at the subsequent slots can have earlier deadlines than some of those who have already arrived. Upper and lower bounds on $P_{PR}(\out)$ are introduced in the following lemma. 
\begin{lemma}
\label{lem:Random_T}
Let $$\I\triangleq \left\{ \sum_{j=0}^{\Tmax}\sum_{i=\Tmin}^{\Tmax}Q_i(n-i-j)>C(\Tmax+1), n\gg 1\right\},$$
    $$\J\triangleq \bigcup_{k=\Tmin}^{\Tmax-1}\left\{\sum_{j=\Tmin}^{k}\sum_{i=\Tmin}^{j}Q_i(n-j)>C(k+1), n\gg 1\right\},$$
$$\Ur\triangleq\I\bigcup\J$$
and $$\Lr\triangleq\bigcup_{k=\Tmin}^{\Tmax}\left\{\sum_{j=\Tmin}^{k}Q_j(n-j)>C(k+1), n\gg 1\right\},$$ then, the events $\Ur$ and $\Lr$ constitute necessary and sufficient conditions on the outage event, respectively. Hence $P(\Lr)\leq P_{PR}(\out)\leq P(\Ur)$.
\end{lemma}
Here also, we assume the system is at steady state.
\begin{proof}
Please refer to Appendix \ref{app:Random_T}.
\end{proof}

\begin{theorem}
\label{th:Random_T}
Let 
\begin{multline*}
\vl_*\triangleq\min_{\Tmin\leq k\leq\Tmax-1}\left\{(k+1)\left[\log\left(\frac{k+1}{\gl \sum_{i=0}^{k-\Tmin}F_{k-i}}\right)-1\right]\right.\\  \left.+\gl \sum_{i=0}^{k-\Tmin}F_{k-i}\right\},
\end{multline*}
the optimal diversity gain of a proactive wireless network with random prediction interval, $d_{PR}(\gamma)$, satisfies
\begin{equation}
\label{eq:dl_Pr_Lower}
\dl_{PR}(\gl)\geq \min\{(\Tmax+1)(\gl-1-\log\gl),\vl_*\}, \quad 0<\gl<1
\end{equation}
for the linear scaling regime, and
\begin{equation}
\label{eq:ds_Pr}
\ds_{PR}(\gs)=(\Tmin+1)(1-\gs),\quad 0<\gs<1,
\end{equation}
for the polynomial scaling regime.
\end{theorem}

\begin{proof}
Please refer to Appendix \ref{app:R_T}.
\end{proof}

Theorem \ref{th:Random_T} determines a lower bound on the optimal prediction diversity gain of the linear scaling regime and fully characterizes the optimal prediction diversity. It is obvious that the lower bound on $\dl_{PR}(\gl)$ depends on the distribution of $T_{q(n)}$, however, this lower bound is always larger than $\dl_{N}(\gl)$ as long as $\Tmax>0$ and $p_{\Tmax>0}$. This can be viewed by considering the term $(\Tmax+1)(\gl-1-\log \gl)$ which is strictly larger than $\dl_{N}(\gl)$ and $\vl_*$ where for any $k$ such that $\Tmin\leq k \leq \Tmax-1$,
\begin{equation*}
\begin{split}
& (k+1)\Biggl[\gl\left(\frac{\sum_{i=0}^{k-\Tmin}{F_{k-i}}}{k+1}\right) -\log\frac{\sum_{i=0}^{k-\Tmin}{F_{k-i}}}{k+1}  -1-\log\gl\Biggr ]\\ & \overset{(a)}{>} (k+1)(\gl-1-\log\gl) \\ & \overset{(b)}{\geq} \dl_{N}(\gl). 
\end{split}
\end{equation*}
Inequality (a) follows as $$0<\frac{\sum_{i=0}^{k-\Tmin}{F_{k-i}}}{k+1}<1$$ and $\gl x-\log x >\gl, \quad \forall x\in(0,1)$, while inequality (b) follows because $k\geq\Tmin\geq 0$. Hence, the proactive network in linear scaling regime with $\Tmax>0$ and $p_{\Tmax}>0$ always improves the diversity gain. 

For the polynomial scaling regime, Theorem \ref{th:Random_T} shows that the prediction diversity gain of a proactive wireless network with random look-ahead interval is dominated by arrivals with $T_{q(n)}=\Tmin$. Hence, the main drawback of this is that, if $\Tmin=0$ the prediction diversity becomes tantamount to that of the non-predictive scenario. However, even though $\Tmin=0$, the outage probability of the predictive network is evaluated numerically in Section \ref{sec:num_res} and shown to outperform the non-predictive case.

\section{Heterogenous QoS Requirements}
\label{sec:Diff_QoS_Users}
We consider two types of users with different QoS requirements, the first is a primary user who has the priority to access the network, whereas the second is a secondary user that is allowed to access the primary network resources opportunistically. That is, it can use the primary resources at any time slot only when there is sufficient capacity to serve all primary requests at that slot with the remaining capacity assigned to the secondary user. This type of opportunistic access to the primary network adds more utilization to its resources while it may get paid by the secondary user for the offered service.


 The primary and secondary requests arrive to the network following two Poisson processes $Q^p(n), n>0$ and $Q^s(n), n>0$ with arrival rates $\lambda^p$ and $\lambda^s$ respectively. We also assume that the network is stable and dominated by primary arrivals as follows.

\begin{assumptions}
\label{assump:1}
\begin{align}
\label{eq:a1}
&\lambda^s+\lambda^p<C,\\
\label{eq:a2}
&\lambda^s<\lambda^p.
\end{align}
\end{assumptions}
The network is reactive to the secondary requests and hence each secondary request will expire if it is not served in the same slot of arrival. In the following subsection, we analyze the performance of the secondary outage probability and diversity gain when the primary network is also reactive, then we proceed to the proactive case.
\subsection{Non-predictive Primary Network}
At the beginning of time slot $n$ the network has $Q^p(n)+Q^s(n)$ arrivals that should be served within the same slot, i.e., all have a deadline of $n$. The network typically serves the primary requests before the secondary. Hence, the diversity gain of the primary network in this scheme, denoted $d_N^p(\gamma^p)$, follows the same expressions obtained in Theorem  \ref{th:diversity_reac}, i.e.,
\begin{align}
& \dl_N^p(\gl^p)=\gl^p-1-\log \gl^p,\quad 0<\gl^p<1\\
& \ds_N^p(\gs^p)=1-\gs^p,\quad 0<\gs^p<1,
\end{align}
where $\bar{\lambda}^p=\gl^p C$ and $\tilde{\lambda}^p=C^{\gs^p}$.

The secondary user, therefore, suffers an outage at time slot $n$ if and only if $$Q^p(n)+Q^s(n)>C, \quad Q^s(n)>0.$$
\begin{theorem}
\label{th:Div_sec_non_pred}
The diversity gain of the secondary network, $d_N^s(\gamma^p,\gamma^s)$, when the primary network is non-predictive, satisfies
\begin{align}
\label{eq:dl_N^s_upper}
&\dl_N^s(\gl^p,\gl^s)\leq \gl^p-1-\log \gl^p,\\
\label{eq:dl_N^s_lower}
&\dl_N^s(\gl^p,\gl^s)\geq \gl^p+\gl^s-1-\log(\gl^p+\gl^s),\\                                                             
\label{eq:ds_N^s}
&\ds_N^s(\gs^p,\gs^s)=1-\gs^p,
\end{align}
where $\bar{\lambda}^s=\gl^s C$, $\tilde{\lambda}^s=C^{\gs^s}$ and $0<\gl^s<\gl^p<1$, $\gl^p+\gl^s<1$ and $0<\gs^s<\gs^p<1$.
\end{theorem}
\begin{proof}
Please refer to Appendix \ref{app:Div_sec_non_pred}.
\end{proof}

Theorem \ref{th:Div_sec_non_pred} reveals that the diversity gain of the secondary user, under non-predictive network, is at most equal to the diversity gain of the primary network in the linear scaling regime and is exactly equal to it in the polynomial scaling regime although the secondary user has strictly less traffic rate than the primary. It can also be noted that $\ds_N^s(\gs^p,\gs^s)$ is independent of $\gs^s$, that is, regardless of how small $\gs^s$ is, the diversity gain of the secondary user is kept fixed at $\ds_N^p(\gs^p)$ as long as $\gs^s>0$. The lower bound in \eqref{eq:dl_N^s_lower}, although does not match the upper bound in \eqref{eq:dl_N^s_upper}, it is always positive and approaches the upper bound when $\gl^s$ is much smaller than $\gl^p$ as shown in Fig. \ref{fig:Upper_lower_diversity_secondary_gs_pO2}.
\begin{figure}
	\centering
		\includegraphics[width=0.45\textwidth]{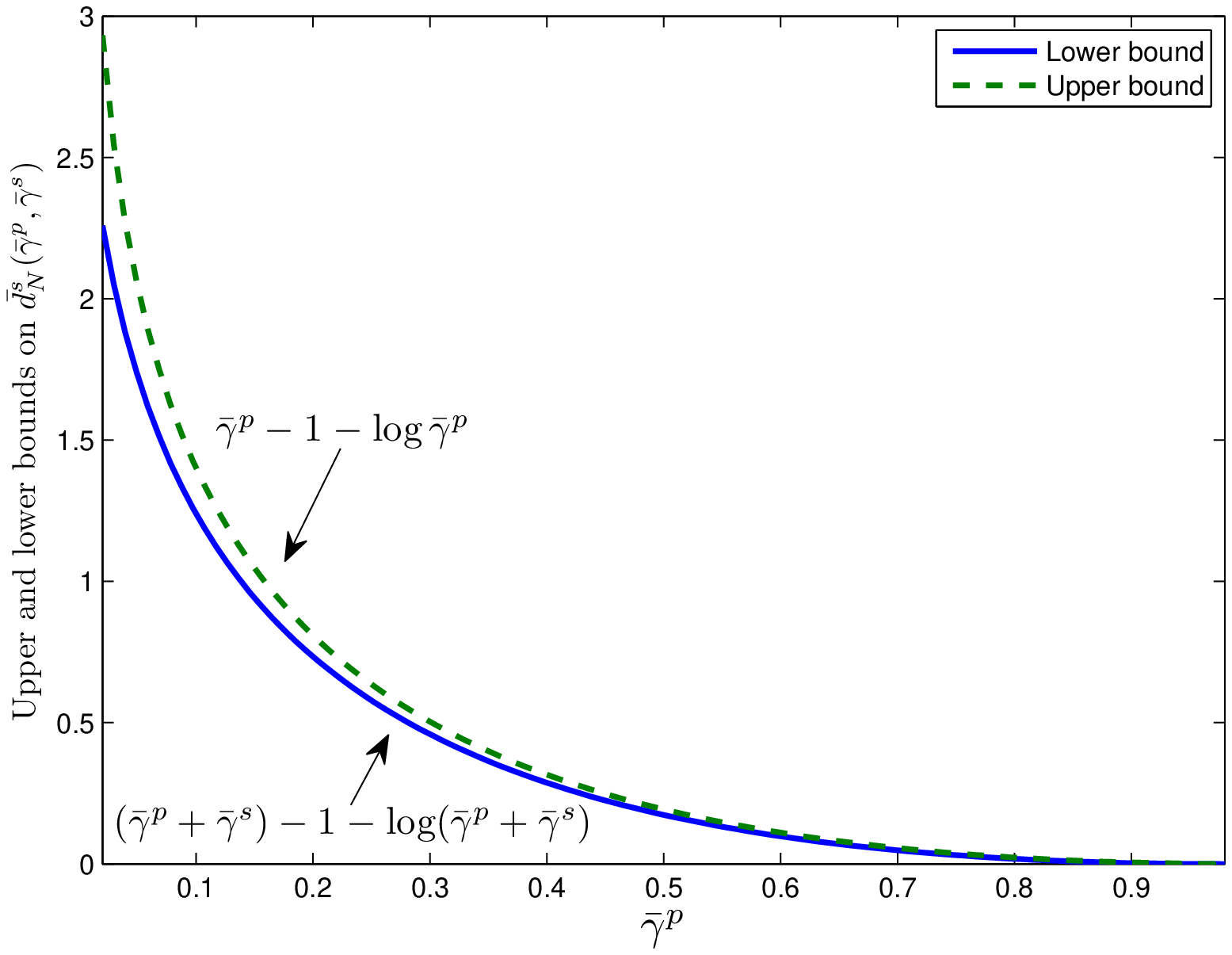}
	\caption{The gap between the upper and lower bounds on $\dl_N^s(\gl^p,\gl^s)$ declines when $\gl^s\ll\gl^p$. In this figure, $\gl^s=0.02$ and $\gl^p\in (\gl^s,1-\gl^s)$. }
	\label{fig:Upper_lower_diversity_secondary_gs_pO2}
\end{figure}

\subsection{Predictive Primary Network}
When the primary network is predictive, the arriving primary requests $Q^p(n)$, $n>0$ are assumed to be predictable with a deterministic look-ahead time $T$. The secondary requests, $Q^s(n)$, conversely, are all urgent.

Let $N_i^p(n)$ be the number of all primary requests awaiting in the network at the beginning of time slot $n$ with deadline $n+i$, $i=0,\cdots, T$ and let $N^p(n)=\sum_{i=0}^{T}N_i^p(n)$. 
\subsubsection{Selfish Primary Scheduling}
By a  \textbf{selfish} primary behavior we mean the primary network has a dedicated capacity $C$ per slot and no secondary request is served at the beginning of time slot $n$ unless all primary requests $N^p(n)$ are served at this slot and $C-N^p(n)>0$. The optimal prediction diversity gain and the outage probability of the primary network in this case are not affected by the presence of the secondary user. On the other hand, the selfish behavior of the primary predictive network cannot improve the outage probability of the secondary user. To show this, let $P_P(\out^s)$ denote the outage probability of the secondary user when the primary network is predictive. Then
\begin{align}
P_P(\out^s)&=P(N^p(n)+Q^s(n)>C, Q^s(n)>0)\nonumber \\
\label{eq:ineq_outage}
           &\geq P(Q^p(n)+Q^s(n)>C, Q^s(n)>0)\\
\label{eq:ineq_outage_2}           
           &=P_N(\out^s),
\end{align}
where inequality \eqref{eq:ineq_outage} follows since $N_T^p(n)=Q^p(n)$ and $N^p(n)\geq N_T^p(n)$. Here we note that the above result holds for any scheduling policy that serves all primary requests in the network at any slot before the secondary requests.

\subsubsection{Cooperative Primary User}

The predictive primary network, however, can act in a \textbf{less-selfish} manner without losing performance and, at the same time, enhance the diversity gain of the secondary user. This can be done by limiting the per-slot capacity dedicated to serve the primary requests in the system. One possible way to do so is to decide the capacity for the primary network dynamically at the beginning of each slot. We suggest the following less-selfish policy.
\begin{definition}
\label{def:D_capacity}
The number of primary requests to be served at slot $n$ is denoted by $C^p(n)$ and given by
\begin{equation}
\label{eq:dynamic_capacity}
C^p(n)\triangleq\min\left\{C, N_0^p(n)+f\times\sum_{i=1}^T N_i^p(n)\right\},
\end{equation}
where $0\leq f\leq 1$, and the primary requests are served according to EDF.
\end{definition}
This scheme determines the maximum number of primary requests that the primary network can serve at the beginning of each slot depending on the number of primary requests with deadline at this slot as well as some factor of the number of other primary requests in the system. Hence, at the beginning of time slot $n$, arriving secondary requests will have the chance to get service if $C-C^p(n)>0$, while the primary network has the capability to schedule the $C^p(n)$ requests according to a service policy that minimizes the primary outage probability (we address the EDF scheduling, however, for simplicity). In the above scheme, if $f$ is chosen to be $1$, the primary network will act selfishly, whereas $f=0$ implies a performance of primary \textbf{non-predictive} network. In the following theorem we show that for some range of $f$, the diversity gain expressions for the primary network satisfy the same bounds of the selfish scenario.

\begin{theorem}
\label{th:dynamic}
Under the dynamic capacity assignment policy in Def. \ref{def:D_capacity} with $f\in [0.5,1]$, the diversity gain of the primary network satisfies
\begin{align}
\dl_P^p(\gl^p)&\geq (T+1)(\gl^p-1-\log\gl^p),  \quad 0<\gl^p<1,\\
\ds_P^s(\gs^p)&=(T+1)(1-\gs^p), \quad 0<\gs^p<1.
\end{align}
\end{theorem}
\begin{proof}
Please refer to Appendix \ref{app:dynamic}.    
\end{proof}

The above theorem thus shows that the predictive primary network satisfies the same diversity gain bounds of the selfish behavior under the proposed dynamic capacity assignment policy as long as $f\in [0.5,1]$. Moreover, it gives a potential for improvement in the outage performance of the secondary users by limiting the number of primary requests served per slot. The outage probability of the secondary network in this case is given by
\begin{align}
P_P(\out^s)&=P\left(Q^s(n)+C^p(n)>C,Q^s(n)>0\right) \nonumber\\
\label{eq:P_P(Os)1}
           &=P\Biggl(Q^s(n)+\min\Bigl\{C,N_0^p(n) \nonumber \\
           &\quad +f\sum_{i=1}^T N_i^p(n)\Bigr\} >C, Q^s(n)>0\Biggr).
\end{align}

To show that even the diversity gain of the secondary network is improved under such policy, we consider the case when $f=0.5$, and $T=1$ for simplicity. In this case, the per-slot capacity of the primary network turns out to be
\begin{equation}
\label{eq:T_1_cap}
C^p(n)=\min\left\{C, N_0^p(n)+0.5 Q^p(n)\right\}
\end{equation}
with
\small
\begin{equation}
\label{eq:N0_evolution}
N_0^p(n+1)=\begin{cases}
            Q^p(n), & \text{if } N_0^p(n)=C,\\
            0.5 Q^p(n)+N_0^p(n)-C, & \text{if } N_0^p(n)<C\text{, }N_0^p(n)\\ & \quad +0.5Q^p(n)\geq C,\\
            0.5 Q^p(n), & \text{if }  N_0^p(n)\\ & \quad+0.5Q^p(n)<C.
        \end{cases}
\end{equation}
\normalsize
It is clear from \eqref{eq:N0_evolution} that
\begin{multline*} 
 P(N_0^p(n+1)=l|N_0^p(n)=i,\cdots, N_0^p(1)=k)\\=P(N_0^p(n+1)=l|N_0^p(n)=i).
\end{multline*}  
  That is, the discrete-time random process $N_0^p(n), n>0$ satisfies the Markov property, and hence, it is a Markov chain. Moreover, it can be easily verified that $N_0^p(n), n>0$ is irreducible and aperiodic as $P(Q^p(n)=q)>0$ for all $q\geq 0$.

The drift of the chain can thus be obtained as
\begin{equation}
\label{eq:drift}
E[N_0(n+1)-N_0(n)|N_0(n)=i]\begin{cases}
														\leq -(1-\gamma^p)C, & \text{if } i\geq C,\\
														\leq \frac{\gamma^p}{2}C, &\text{if } i<C.
                            \end{cases}
\end{equation}
Then, by Foster's theorem \cite{Foster}, the Markov chain is positive recurrent, and hence has a stationary state distribution.

\begin{theorem}
\label{th:div_sec_dyn}
Suppose that the system is operating at the stationary distribution of ${N_0^p(n),n>0}$, the diversity gain of the secondary network, $d_P^s(\gamma^p,\gamma^s)$, under the dynamic capacity allocation for the primary satisfies
\begin{equation}
\label{eq:dPs_lin}
\dl_P^s(\gl^p,\gl^s)\geq -\gl^s(\yl^2-1)-2\gl^p(\yl-1)+2\log(\yl),
\end{equation}
where $$\yl=-\frac{\gl^p}{2\gl^s}+\frac{\sqrt{(4\gl^s+{\gl^p}^2)}}{2\gl^s}$$ and
\begin{equation}
\label{eq:dPs_poly}
\ds_P^s(\gs^p,\gs^s)\geq \begin{cases}
           (1-\gs^p), & 1+\gs^s\geq 2\gs^p,\\
            \frac{1}{2}(1-\gs^s), & 1+\gs^s< 2\gs^p.
        \end{cases}
\end{equation}
\end{theorem}
\begin{proof}
Please refer to Appendix \ref{app:div_sec_dyn}.
\end{proof}

The right hand side of inequality \eqref{eq:dPs_lin} will be shown in Section \ref{sec:num_res} to be strictly larger than the right hand side of \eqref {eq:dl_N^s_upper} for a range of $\gl^s$, which implies a strict improvement in the diversity gain of the secondary network without any loss in the diversity gain of the primary. However, the right hand side of inequality \eqref{eq:dPs_poly} shows that if $1+\gs^s<2\gs^p$, then the diversity gain of the secondary network is at least equal to its non-predictive counterpart.

\section{Robustness to Prediction Errors}
\label{sec:Robustness}
In the previous sections we have assumed that the prediction mechanism is error free, that is, all predicted requests are true and will arrive in future after exactly the same look-ahead period of prediction. Under this assumption, we managed to treat the predicted arrival process with deterministic look-ahead time as a delayed version of the actual arrival process. However, in practical scenarios, this is not necessarily the case. In this section we provide a model for the imperfect prediction process and investigate its effect on the prediction diversity gain with fixed look-ahead interval $T$ assuming a single class of QoS.

Let $Q(n)$, $n>0$ be the actual arrival process that the network should predict $T$ slots ahead. This process, as introduced in Section \ref{sec:sys_model}, is Poisson with rate $\lambda$. Because the prediction mechanism employed by the network may cause errors, the predicted arrival process differs from the actual arrival process. The prediction mechanism is supposed to cause two types of errors:
\begin{enumerate}
	\item It predicts false requests, those will not arrive actually in future, and serves them, resulting in a waste of resources.
	\item It fails to predict requests and, as a consequence, the network encounters urgent arrivals (unpredicted requests that should be served in the same slot of arrival).
\end{enumerate} 

So, we model the predicted process as 
\begin{equation}
Q^{E}(n)=Q'(n)+Q''(n)
\end{equation}
where $Q'(n)$, $n>0$ is the arrival process of the predicted requests. It represents the number of arriving requests at the beginning of time slot $n$ with deadline $n+T$. The process $Q''(n)$, $n>0$ represents the number of unpredicted requests that arrive at the beginning of time slot $n$ and must be served in the same slot because the network has failed to predict them. We assume further that $Q'(n)$ and $Q''(n)$ are independently Poisson distributed with arrival rates $\lambda'$ and $\lambda''$, respectively.

Since $Q''(n)$ is a part of the requests $Q(n)$, then
\begin{equation}
\label{eq:Cond1}
0\leq\lambda''<\lambda
\end{equation}
where the second inequality is strict because we assume that $Q'(n)$ contains truly predicted requests as well as mistakenly predicted requests, which also implies
\begin{equation}
\label{eq:Cond2}
\lambda'+\lambda''\geq \lambda
\end{equation}
Moreover, the network is stable as long as
\begin{equation}
\label{eq:Cond3}
\lambda'+\lambda''<C.
\end{equation}

For the linear scaling regime, the arrival processes $\Ql'(n)$ and $\Ql''(n)$, $n>0$ have arrival rates $\al'\gl C $ and $\al''\gl C$ respectively.
Applying conditions \eqref{eq:Cond1}-\eqref{eq:Cond3} to $\al'\gl C $ and $\al''\gl C$ we obtain
\begin{equation}
\al''< 1
\end{equation}
and
\begin{equation}
1\leq \al'+\al''<\frac{1}{\gl}
\end{equation}
So, if the prediction mechanism is perfect, then $\al'=1$ whereas $\al''=0$.

The arrival process $\Ql^E(n)$, $n>0$, can be considered as a predicted process with random look-ahead interval that takes on values $0$ and $T$. Hence, using the event $\Ur$ defined in Lemma \ref{lem:Random_T}, we obtain the following lower bound on the prediction diversity gain, $\dl_{P}^E(\gl)$,
\begin{multline}
\label{eq:dl_P^E}
\dl_{P}^{E}(\gl)\geq \min\left\{(T+1)\left[(\al'+\al'')\gl-1-\log\left(\gl(\al'+\al'')\right)\right],\right. \\ \left.
\al''\gl-1-\log(\al''\gl)\right\}
\end{multline}
The best operating point (prediction window) that maximizes the right hand side of \eqref{eq:dl_P^E} is when both terms in the $\min\{.\}$ are equal, which implies
\begin{equation}
\label{eq:Tl}
\bar{T}_{crit}=\frac{\al''\gl-1-\log(\al''\gl)}{(\al'+\al'')\gl-1-\log(\gl (\al'+\al''))}.
\end{equation}

Since $\al''<1$, then for $\bar{T}_{crit}$ derived in \eqref{eq:Tl}, we obtain $\dl_P^E(\gl)>\dl_N(\gl)$.

For the polynomial scaling regime, the processes $\Qs'(n)$ and  $\Qs''(n)$, $n>0$ have arrival rates $C^{\as'\gs}$ and $C^{\as''\gs}$ respectively. Applying conditions \eqref{eq:Cond1}-\eqref{eq:Cond3} to the arrival rates $C^{\as'\gs}$ and $C^{\as''\gs}$, we obtain,
\begin{equation}
\label{eq:Cond1_s}
\as''<1,
\end{equation}
\begin{equation}
\label{eq:Cond2_s}
C^{\as'\gs}+C^{\as''\gs}\geq C^{\gs},
\end{equation}
and
\begin{equation}
\label{eq:Cond3_s}
C^{\as'\gs}+C^{\as''\gs}<C.
\end{equation}
If the prediction mechanism is perfect, then $\as'=1$ whereas $\as''=-\infty$.

We also use events $\Ur$ and $\Lr$ from Lemma \ref{lem:Random_T} to determine the prediction diversity gain with imperfect prediction mechanism, $\ds_{P}^E(\gs)$, as
\begin{equation}
\ds_{P}^E(\gs)=\min\{(T+1)\left[1-\max\{\as',\as''\}\gs\right], 1-\as''\gs\}.
\end{equation}
Nevertheless, since at $\ds_P^E(\gs)$ is at $C\rightarrow\infty$, then from \eqref{eq:Cond2_s}, \eqref{eq:Cond3_s}, as $C\rightarrow\infty$, we obtain, $1\leq\as'<\frac{1}{\gs}$. And from \eqref{eq:Cond1_s}, $\max\{\as',\as''\}=\as'$.
Hence,
\begin{equation}
\ds_{P}^E(\gs)=\min\{(T+1)(1-\as'\gs), 1-\as''\gs\}.
\end{equation} 

So, to obtain the maximum diversity gain, the best prediction window $\tilde{T}_{crit}$ should satisfy
\begin{equation}
\tilde{T}_{crit}=\frac{(\as'-\as'')\gs}{1-\as''\gs},
\end{equation}
and at this point, since $\as''<1$, we have $\ds_{P}^E(\gs)>\ds_{N}(\gs)$.

This section hence has shown theoretically that even under imperfect prediction mechanisms, the prediction window size is judiciously chosen to strike the best balance between the predicted traffic and the urgent one. 

\section{Proactive Scheduling in Multicast Networks}
\label{sec:Multicast}
This section sheds light on the predictive multicast networks and investigates the diversity gains that can be leveraged from efficient scheduling of multicast traffic. Typically, multicast traffic minimizes the usage of the network resources because the same data is sent to a group of users consuming the same amount of resources that serve only a single user which is taken to be unity \cite{Multicast}. So, even in the non-predictive case, the multicast traffic is expected to result in an improved diversity gain performance over its unicast counterpart, discussed in the previous sections.

Furthermore, when the multicast traffic is predictable, there is an additional gain that can be obtained from the ability to \emph{align} the traffic in time. That is, the network can keep on receiving predictable requests that target the same data over time then serve them altogether as the earliest deadline approaches. In this case, the network will end up serving all the gathered requests in a window of time slots with the same resources required to serve one request, and hence will significantly improve the diversity gain of the network. We assume that there are $L$ data sources available (e.g. files, packets, movies, podcasts, etc.) for multicast transmission. The number of multicast requests arriving at the beginning of time slot $n>0$ is a random variable $Q^m(n)$ which is assumed to be Poisson distributed with mean $\lambda^m$.

Assuming that the data sources are demanded independently across time and requests, the process $Q^m(n), n>0$ can be decomposed into
\begin{equation*}
Q^m(n)=\sum_{l=1}^{L}{Q^{m,[l]}(n)}, \quad \textrm{for all } n>0,
\end{equation*}
where $Q^{m,[l]}(n)$ denotes the number of multicast requests for data source $l\in\{1,\cdots, L\}$ arriving in slot $n$, and is Poisson distributed with mean $\lambda^{m,[l]} \triangleq p^{[l]} \lambda^m,$ where $\vp \triangleq (p^{[l]})_{l=1}^L$ is a valid probability distribution\footnote{$\vp$ is a valid distribution if $0\leq p^{[l]}\leq 1$ and $\sum_{l=1}^L p^{[l]}=1$.} capturing the potentially asymmetric multicast demands over the pool of $L$ data sources.

In this section we focus only on the analysis of the linear scaling regime where the potential improvement in the diversity gain is tangible \footnote{The additional multicast gains do not appear in the polynomial scaling regime because the traffic to each data source vanishes asymptotically, as $C\to\infty$, when the number of data sources $L$ scales with $C$, implying that at most one request can target a data source at each slot, i.e., the multicast traffic will approach the unicast as $C\to\infty$.}. The mean number of arriving multicast requests scales with $C$ as $\lambda^m=\gm C$, $\gm\in(0,1)$. The number of data sources $L$ scales also linearly with $C$ as $L=\thl C$, $\thl>0$. 

The binary parameter $X^{m,[l]}(n)$ for each multicast data source $l\in\{1,\cdots,L\}$ is defined as
  \begin{multline}
X^{m,[l]}(n)\triangleq\left\{
\begin{array}{l l}
1, \quad \text{if } Q^{m,[l]}(n)>0,\\
0, \quad \text{if } Q^{m,[l]}(n)=0,
\end{array}
\right.
 l=1, \cdots, L,
\end{multline}
which gives the indicator of at least one multicast request for data source $l$ arrives at slot $n.$ And, under the aforementioned Poisson assumptions, $X^{m,[l]}(n)$ is a simple Bernoulli random variable with parameter
\begin{equation}
A^{m,[l]}=1-e^{-p^{m,[l]}\lambda^m}, \quad l\in\{1,\cdots,L\}.
\end{equation}

We denote the total number of distinct multicast data requests arriving in slot $n$ as $S^m(n),$ defined as
\begin{equation}
\label{eq:S(n)}
S^m(n)\triangleq \sum_{l=1}^{L}X^{m,[l]}(n).
\end{equation}

\begin{definition}
Let $N^{m,[l]}_0(n)$ denote the indicator that there is \emph{at least} one awaiting multicast request for data source $l\in\{1,\cdots,L\}$ that expires in slot $n.$
Then, letting $N^m_0(n)\triangleq \sum_{l=1}^L N^{m,[l]}_0(n),$ the multicast outage event is defined as
\begin{equation*}
{\outm} \triangleq
    \left\{N^m_0(n) >C, n\gg 1 \right\}.
\end{equation*}
\end{definition}

The pure multicast network will be investigated in the following subsection where the diversity gain of its non-predictive side will be shown to be larger than its unicast counterpart, furthermore, the alignment property of the predictive multicast will be proven to result in a significantly improved diversity gain, that scales super-linearly with the prediction interval $T$. Then, the subsequent subsection will address a composite network consisting of unicast and multicast traffics; the potential diversity gain will be investigated under different prediction scenarios.

\subsection{Symmetric Multicast Demands}
Suppose that the number of data sources scales with $C$ as $L= \thl C$, $\thl>0$. Then, $\thl\leq 1$ implies zero outage probability and infinite diversity gain regardless of the value of $\gm$. This is the first gain improvement that can be leveraged from the nature of the multicast traffic. We now confine the analysis to the case when $\thl>1$. Assume that the multicast demands are equally distributed on the available data sources, i.e.
\begin{equation*}
p^{[l]}=p=\frac{1}{\thl C},
\end{equation*}
\begin{equation*}
A^{m,[l]}=A^m=1-e^{-\frac{\gm}{\thl}}, \quad \forall l\in\{1,\cdots,L\}.
\end{equation*}

\subsubsection {Non-predictive Multicast Network}
Under the above symmetric setup (and assuming $\thl>1$), the random variable $S^m(n)$ turns out to have a binomial distribution with parameter  $A^m$ and the outage probability in this case, denoted by $P_N(\outm)$, is equal to $P(S^m(n)>C)$. In other words, the multicast outage occurs in slot $n$ if and only if the number of distinct data sources requested at this slot is larger than the network capacity.
\begin{theorem}
\label{th:Div_multi_rea}
The diversity gain of non-predictive multicasting, denoted by $\dl_N(\gm,\thl)$, is given by
\begin{multline}
\label{eq:Div_multi_rea_lin}
\dl_N(\gm,\thl)=(\thl-1)\log(\thl-1)-\thl \log \thl+\gm\left(\frac{\thl-1}{\thl}\right)\\ -\log\left(1-e^{-\frac{\gm}{\thl}}\right),\quad 0<\gm<1, \quad \thl>1.
\end{multline}
\end{theorem} 
\begin{proof}
Please refer to Appendix \ref{app:Div_multi_rea}.
\end{proof}

\begin{figure}
	\centering
		\includegraphics[width=0.45\textwidth]{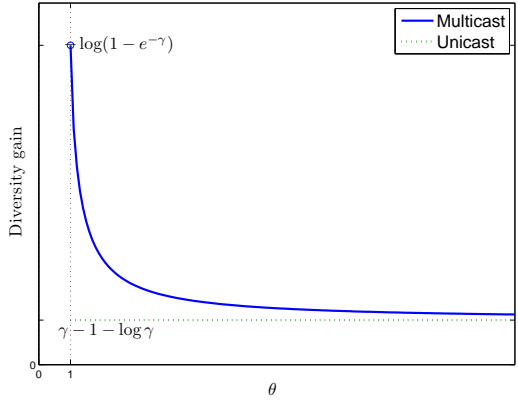}
	\caption{Diversity gain of the non-predictive multicast network monotonically decreases with $\thl$. However, it is lower bounded by the diversity gain of non-predictive unicast networks.}
	\label{fig:Paper_div_multi_linear_vs_theta}
\end{figure}

 Theorem \ref{th:Div_multi_rea} and Fig. \ref{fig:Paper_div_multi_linear_vs_theta}, which depicts the diversity gains of non-predictive multicast \eqref{eq:Div_multi_rea_lin} and unicast \eqref{eq:div_rea_lin} networks with $\gm=\gl$, show that $\dl_N(\gm,\thl)$ is monotonically decreasing in $\thl$. As $\thl$ increases, the number of data sources in the network grows faster with $C$, and hence, from \eqref{eq:Div_multi_rea_lin},
\begin{align}
\lim_{\thl\rightarrow\infty} \dl_N(\gm,\thl)&=\gm-\log\gm-1=\dl_N(\gm).
\end{align}


That is, multicast diversity gain $\dl_N(\gm,\thl)$ is strictly greater than its unicast counterpart $\dl_N(\gm)$, and converges to it in the limit as $\thl \rightarrow \infty$. In fact, a much stronger result is that, when $\gm=\gl$,
\begin{equation}
\begin{split}
\lim_{\thl\rightarrow\infty} L A^m&=\lim_{\thl\rightarrow\infty}\thl C \left(1-e^{-\frac{\gm}{\thl}}\right)\\
                                  &=\gm C,\quad 0<\gm<1,
\end{split}
\end{equation}
we have also $A^m\rightarrow 0$ and $L=\thl C\rightarrow\infty$ as $\thl\rightarrow\infty$.
Therefore, $S^m(n)$ converges in distribution to $\Ql(n)$, and consequently, $P_N(\outm)\rightarrow P_{N}(\out)$, $\thl\rightarrow\infty$.

In this subsection, we have highlighted the extra diversity gain achieved through one of the multicast properties, that is all the requests arriving to the network at time slot $n$ and demanding a certain data source are all served with one unit resources exactly as if only one request demands that data source.

\subsubsection{Predictive Multicast Network}

Now suppose that the symmetric multicast network has predictable demands with a prediction window of $T>0$ slots. The traffic alignment in this case appears in the following sense, the resource serving a group of requests arriving at slot $n$ also serves all other requests in the system (that have arrived withing the previous $T$ slots) requesting the same data source. So, the resource value is extendable across time. The prediction capability of the network is thus equal to infinity as long as $\thl\leq (T+1)$, which implies a multiplicative gain of $T+1$ in the number of data sources that the network can support with zero outage probability, as compared to the non-predictive case.

Consider then the other range of $\thl$, that is $\thl>(T+1)$. The network now is subject to outage events and efficient scheduler has to be employed. Because of the symmetric demands, we focus the analysis on the EDF scheduling. Let the optimal diversity gain in this predictive scenario be denoted by $\dl_P(\gm,\thl)$, in \cite{ISIT}, we have shown that $\dl_P(\gm,\thl)\geq (T+1) \dl_N(\gm,\thl)$ which is consistent with the results of Subsection \ref{subsec:Pred_Div_Gain} as the predictability multiplies the diversity gain by a factor of at least $T+1$. However, we show now that the alignment property can even improve the diversity gain and result in a super-linear scaling of $\dl_P(\gm,\thl)$ with $T$.

\begin{theorem}
\label{th:div_multi_pred}
The optimal diversity gain of the predictive multicast network with symmetric demands, $\dl_P(\gm,\thl)$, satisfies
\begin{equation}
\label{eq:div_multi_pro}
\begin{split}
\dl_P(\gm,\thl)&\geq (T+1)\log\left(\frac{(1-\xm)(T+1)}{\xm(\thl-(T+1))}\right)\\
               &\quad     -\thl\log\left(1-\xm+\frac{(1-\xm)(T+1)}{\thl-(T+1)}\right),\\
               &\triangleq \mathbb{L}_{sym}.     
\end{split}
\end{equation}
where $$\xm=1-\exp\left({-\frac{(T+1)\gm}{\thl}}\right).$$
\end{theorem}
\begin{proof}
Please refer to Appendix \ref{app:div_multi_pred}.
\end{proof}
The new lower bound $\mathbb{L}_{sym}$ takes into account the alignment property of the predictable multicast traffic, and thus shows significant increase in the diversity gain with $T$ as compared to the older bound $(T+1)\dl_N(\gm,\thl)$ in Fig. \ref{fig:Journal_newer_lower_bound_vs_T}.
\begin{figure}
	\centering
		\includegraphics[width=0.45\textwidth]{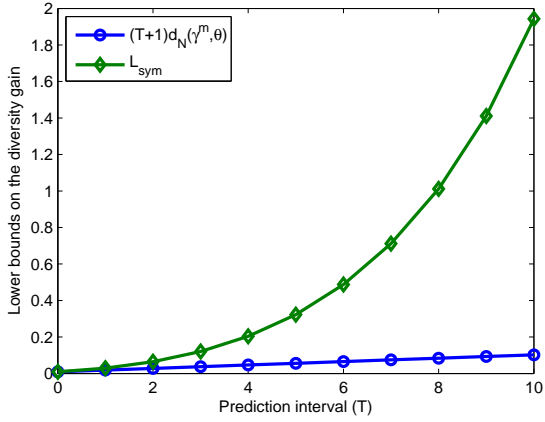}
	\caption{Superlinear increase in the diversity gain of the multicast network with the prediction interval $T$ because of the alignment property. In this figure $\gm=0.9$ and $\thl=15$.}
	\label{fig:Journal_newer_lower_bound_vs_T}
\end{figure}
 
\subsection{Multicast and Unicast Traffic}
Generally, wireless networks support both types of traffic: multicast and unicast. For instance, a smart phone user my receive unicast data such as e-mail or electronic bank statement as well as multicast data such as movies or podcasts. In this subsection we investigate the potential diversity gain of wireless networks encompassing both types of traffic under different predictability conditions.

The multicast traffic model adopted here is exactly as defined in the beginning of this section, with the only difference is we assume that $L=\thl C$, where $\thl\in(0,1)$. The multicast data sources are also equally demanded, each with probability $$A^m=1-\exp\left(-\frac{\gm}{\thl}\right).$$ The unicast traffic arrives at the beginning of each slot $n$ according to $Q^u(n)$ which is Poisson distributed with mean $\lambda^u=\gu C$, $\gu\in(0,1)$. Each of the unicast requests consumes one unit of the system capacity. The stability condition of the non-predictive network necessitates that
\begin{equation}
\label{eq:stability}
A^m\thl + \gu <1.
\end{equation} 
\begin{definition}
Letting $N^u_0(n)$ denote the number of unicast requests in the system at the beginning of time slot $n$, the combined outage event of the wireless network with unicast and multicast traffic is defined as
\begin{equation*}
\outa=\left\{N^m_0(n)+N^u_0(n)>C, n\gg 1\right\}.
\end{equation*} 
\end{definition}

In \cite{ISIT}, we have addressed the case when only on multicast data source exists in the network an consumes $\mu C$, $\mu\in(0,1)$ of the available resources to supply data. This data source shares the network with unicast traffic. We have shown the impact of the multicast traffic alignment on the diversity gain where more gains can be leveraged by gathering more of the predictable multicast traffic and serving them altogether in a single slot. Alternatively, in this subsection we address the scenario of multiple data sources each consumes one unit of the available resources. We will investigate the diversity gain of the network in the following four scenarios of demand predictability:
\begin{enumerate}
\item Both unicast and multicast traffics are non-predictive.
\item Unicast is non-predictive but multicast is predictive.
\item Both unicast and multicast traffics are predictive.
\item Unicast is predictive but multicast is non-predictive.
\end{enumerate}   

\subsubsection{Scenario 1: Both Unicast and Multicast Traffics are Non-predictive}
In this scenario, all of the arriving requests are urgent and hence, $N^m_0(n)=S^m(n)$ and $N^u_0(n)=Q^u(n)$.
\begin{theorem}
\label{th:S1}
Let the outage probability in Scenario 1 be denoted by $P_1(\outa)$ and the associated diversity gain be denoted by $\dl_1(\gu,\gm,\thl)$, then
\begin{equation}
\label{eq:S1}
\begin{split}
\dl_1(\gu,\gm,\thl)=&\log(y_1)+\gamma^u(1-y_1)\\
                           &-\thl\log\left(e^{-\frac{\gamma^m}{\thl}}+y_1\left(1-e^{-\frac{\gamma^m}{\thl}}\right)\right),
\end{split}
\end{equation}
where
\begin{equation*}
\begin{split}
y_1=&\frac{1}{2\gu \left(e^{\frac{\gm}{\thl}}-1\right)}\Biggl[\Bigl((\thl^2-2\thl+1)e^{\frac{2\gm}{\thl}}\\
   &+\bigl(-2\thl^2+2\thl(\gu+2)+2(\gu-2)\bigr)e^{\frac{\gm}{\thl}}+\thl^2\\
   &-2\thl(\gu+1)+{\gu}^2-2\gu+1\Bigr)^{\frac{1}{2}}+(1-\thl)e^{\frac{\gm}{\thl}}\\
   &+\thl-\gu-1\Biggr].
\end{split}
\end{equation*}
\end{theorem}
\begin{proof}
Please refer to Appendix \ref{app:S1}.
\end{proof}
Theorem \ref{th:S1} thus tightly characterizes the diversity gain of the network in Scenario 1. The expression of $\dl_1(\gu,\gm,\thl)$, however, is not insightful, so it will be compared graphically to the results of the other scenarios.

\subsubsection{Scenario 2: Unicast is Non-predictive but Multicast is Predictive}
In this scenario, the network can predict the multicast requests $T$ slots ahead, whereas the unicast traffic is urgent.
We consider a scheduling policy $\pi_2$ to establish a lower bound on the optimal diversity gain, denoted $\dl_2(\gu,\gm,\thl)$, of this scenario.
\begin{definition}[Scheduling Policy $\pi_2$]
At each slot $n$, the scheduling policy $\pi_2$ serves as much as possible of the existing requests in the system in the following order:
\begin{enumerate}
\item Multicast data sources demanded by urgent requests, $N^u_0(n)$.
\item Unicast requests, $Q^u(n)$.
\item The rest of the multicast data sources according to EDF. 
\end{enumerate}  
\end{definition}
The policy $\pi_2$ is a slightly modified version of EDF with priority given to urgent multicast requests.
\begin{theorem}
\label{th:S2}
Let the outage probability in Scenario 2 under the scheduling policy $\pi_2$ be denoted $P_2(\outa)$ and the optimal diversity gain be denoted by $\dl_2(\gu,\gm,\thl)$, then
\begin{equation}
\label{eq:S2_l2}
\begin{split}
\dl_2(\gu,\gm,\thl)&\leq  \min\Bigl\{\dl_N(\gu),(T+1)\log y_2\\
                   &\quad-(T+1)\gu(y_2-1)\\
                   &\quad-\thl\log(1-\xm+\xm y_2)\Bigr\},\\
                   & \triangleq \mathbb{L}_2.     
\end{split}
\end{equation}
and
\begin{equation}
\label{eq:S2_u2}
\dl_2(\gu,\gm,\thl)\leq \dl_N(\gu)\triangleq \mathbb{U}_2,
\end{equation}
where $\dl_N(\gu)$ is as derived in \eqref{eq:div_rea_lin} with $\gl=\gu$, and
\begin{equation*}
\begin{split}
y_2=&\frac{1}{2\xm\gu (T+1)}\Biggl[\Biggl(\Bigl((1-\xm)^2{\gu}^2+2\xm\gu(1-\xm)\\
    & +{\xm}^2\Bigr)^2T^2 +\Bigl([2\xm\gu(1-\xm)-2{\xm}^2]\thl\\
    &+2{\xm}^2(1-\xm)^2+4\xm\gu(1-\xm)+2{\xm}^2\Bigr)T\\
    &+[2\xm\thl(1-\xm)-2{\xm}^2]\thl+{\gu}^2(1-\xm)^2\\
    &+2\xm\thl(1-\xm)+{\xm}^2(1+\thl)^2\Biggr)^{\frac{1}{2}}\\
    &+\Bigl((\xm-1)\gu\Bigr)T-\xm\thl+\gu(\xm-1)+\xm\Biggr].
\end{split}
\end{equation*}
\end{theorem}
\begin{proof}
Please refer to Appendix \ref{app:S2}.
\end{proof}

The upper and lower bounds on $\dl_2(\gu,\gm,\thl)$ established in Theorem \ref{th:S2} match each other as $T$ increases. In fact, the second term in $\min\{.,.\}$ of expression \eqref{eq:S2_l2} is monotonically increasing in $T$, and hence $\exists t$ such that $T\geq t$ implies $\dl_2(\gu,\gm,\thl)=\dl_N(\gu)$. This result means that, efficient scheduling of the predictable multicast traffic results in the same diversity gain that will be obtained if the system sees only the unicast traffic.
\begin{figure}
	\centering
		\includegraphics[width=0.42\textwidth]{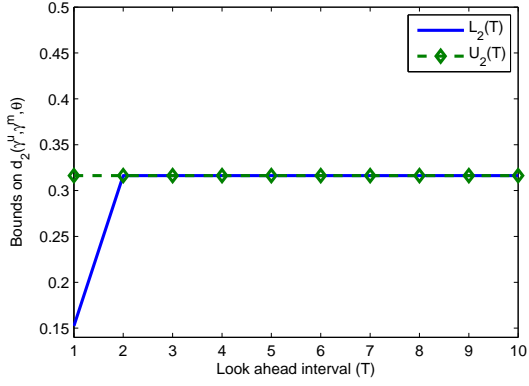}
	\caption{As $T$ increases, the system attains the same diversity gain of the non-predictive unicast network. In this figure, $\thl=0.7$, $\gm=0.9$ and $\gu=0.4$.}
	\label{fig:Journal_bounds_on_d2_vs_T}
\end{figure}
This result is clarified in Fig. \ref{fig:Journal_bounds_on_d2_vs_T} where the lower bound $\mathbb{L}_2$ increases in $T$ until it becomes dominated by $\dl_N(\gu)$ at $T=2$, and from this point on, $\mathbb{L}_2$ and $\mathbb{U}_2$ coincide and the diversity gain of the network is only determined by the non-predictive unicast traffic.    
\subsubsection{Scenario 3: Both Unicast and Multicast Traffics are Predictive}
In this scenario we assume that both traffics are predictable with the same look-ahead interval of $T$ slots. The scheduling policy we consider is EDF where requests are served in the order of their arrival.
\begin{theorem}
\label{th:S3}
Let the outage probability of the network in Scenario 2 under EDF scheduling policy be denoted by $P_3(\outa)$ and the optimal diversity gain of this scenario be denoted by $\dl_3(\gu,\gm,\thl)$, then
\begin{equation}
\label{eq:S3_l3}
\begin{split}
\dl_3(\gu,\gm,\thl)& \geq (T+1)\log y_2 -(T+1)\gu(y_2-1)\\
                   &\quad-\thl\log(1-\xm+\xm y_2)\\
                   &\triangleq \mathbb{L}_3.
\end{split}
\end{equation} 
\end{theorem}
\begin{proof}
Please refer to Appendix \ref{app:S3}.
\end{proof}
In Scenario 3 one should expect that the optimal diversity gain should be the largest amongst the other three scenarios. To highlight this intuition, an upper bound will be established on the diversity gain of Scenario 4.

\subsubsection{Scenario 4: Unicast is Predictive but Multicast is Non-predictive}
Assuming that the unicast traffic is predictable with a look-ahead window of $T$ slots, and the multicast traffic is urgent.
\begin{theorem}
\label{th:S4}
Let the optimal diversity gain of Scenario 4 be denoted by $\dl_4(\gu,\gm,\thl)$ and the minimum possible outage probability be denoted by $P^*_4(\outa)$, then
\begin{equation}
\label{eq:S4_u4}
\begin{split}
\dl_4(\gu,\gm,\thl)& \leq \dl_1(\gu,\gm,\thl)+T\Bigl[2\log y_4-\gu(y_4-1)\\
                   &\quad -2\thl\log(1-A^m+A^m y_4)\Bigr]\\
                   &\triangleq \mathbb{U}_4, 
\end{split}
\end{equation}
where
\begin{equation*}
\begin{split}
y_4=\frac{1}{2\gu A^m}\Biggl[& \Bigl ((4\thl^2-4\thl(\gu+2)+(2-\gu)^2){A^m}^2+{\gu}^2\\
    &+(4\gu\thl-2{\gu}^2+4\gu)A^m\Bigr)^{\frac{1}{2}}\\
& +(-2\thl+\gu+2)A^m-\gu\Biggr].
\end{split}
\end{equation*}
\end{theorem}
\begin{proof}
Please refer to Appendix \ref{app:S4}.
\end{proof}
\begin{figure}
	\centering
		\includegraphics[width=0.42\textwidth]{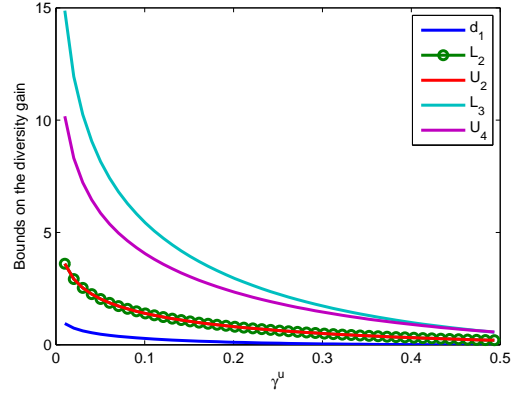}
	\caption{Bounds on the optimal diversity gain versus the unicast traffic factor $\gu$. In this figure, $\gm=0.9$, $\thl=0.7$ and $T=4$ for any predictive network.}
	\label{fig:Journal_scenarios_bounds}
\end{figure}

\begin{figure}
	\centering
		\includegraphics[width=0.42\textwidth]{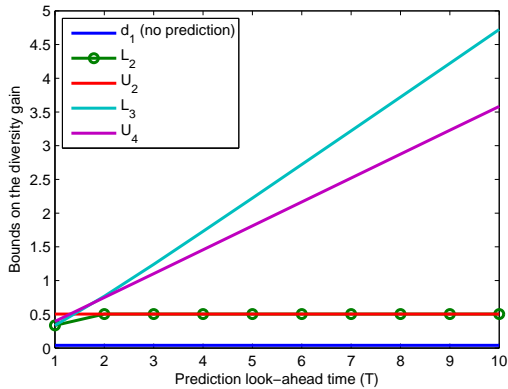}
	\caption{Bounds on the optimal diversity gain versus the prediction look-ahead time $T$. In this figure, $\gu=0.4$, $\gm=0.9$ and $\thl=0.7$.}
	\label{fig:Journal_scenarios_vs_T}
\end{figure}

To collectively compare the obtained bounds on the optimal diversity gain of the discussed scenarios, Fig.\ref{fig:Journal_scenarios_bounds} plots the different bounds obtained in the last four theorems versus $\gu$, where the range of $\gu$ ensures that \eqref{eq:stability} is satisfied, and hence the non-predictive network always sees a positive diversity gain.
It is clear from the figure that the totally predictive network (of Scenario 3) has the highest possible diversity gain as the lower bound $\mathbb{L}_3$ even exceeds the upper bound $\mathbb{U}_4$ on the entire range of plotted $\gu$. Also, it shows that $\mathbb{L}_2$ and $\mathbb{U}_2$ are coinciding at $\dl_N(\gu)$, and of course this is the best diversity gain that the network can achieve with unpredictable unicast traffic.

Also, Fig. \ref{fig:Journal_scenarios_vs_T} demonstrates the effect of the prediction look-ahead period $T$ on the derived bounds. It shows that both $\mathbb{L}_3$ and $\mathbb{U}_4$ are both increasing in $T$, and that as $T$ increases $\mathbb{L}_3$ exceeds $\mathbb{U}_4$ and $\mathbb{L}_2$ matches $\mathbb{U}_2$.


\section{Simulation Results}
\label{sec:num_res}
The analytical results obtained in this paper are demonstrated through numerical simulations in this section. The outage probability is quantified as the ratio of the number of slots that suffer expired requests to the total number of simulated slots. Each simulation result is obtained by averaging a $100$ sample paths each contains a $1000$ slots.  
\subsection{Diversity Gain of Deterministic and Random $T$ Scenarios}
\begin{figure}[ht]
  \centering
  \subfloat[Linear scaling regime: $\gl^p=0.8$.]{\label{fig:Paper_FR_T_Linear}\includegraphics[width=0.4\textwidth]{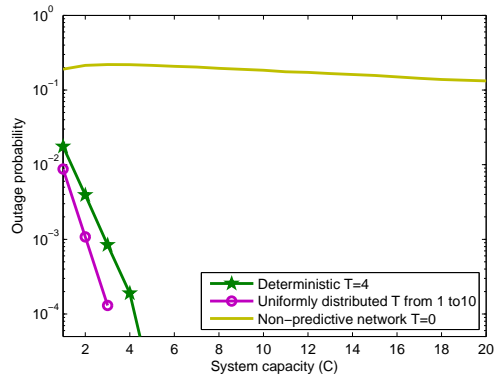}}\\   \subfloat[Polynomial scaling regime: $\gs^p=0.8$.]{\label{fig:Paper_sim_FR_T}\includegraphics[width=0.4\textwidth]{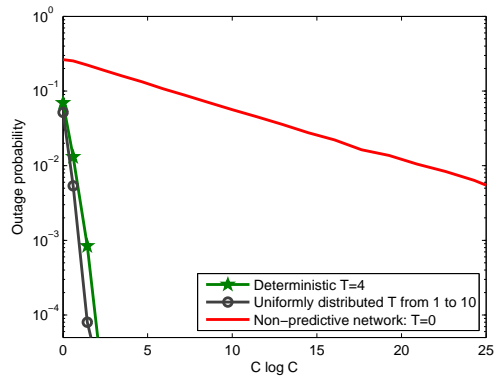}} 
  \caption{Outage probability is significantly improved by proactive networks.}
  \label{fig:Outage_proactive}
\end{figure} 

Fig. \ref{fig:Outage_proactive} compares the outage probability of proactive networks with different look-ahead schemes to the non-predictive network. The results obtained for the linear scaling regime are plotted versus $C$ in Fig. \ref{fig:Paper_FR_T_Linear} and for the polynomial scaling regime are plotted versus $C\log C$ in Fig. \ref{fig:Paper_sim_FR_T}. It is obvious from both figures that being proactive significantly enhances the outage probability performance at a given capacity, or it considerably reduces the required capacity to satisfy a given level of outage performance. This ascribes to the more flexibility given to the predictive network that allows it to schedule the arriving requests over a longer time horizon compared to the urgent demand of the non-predictive network. The effect of the distribution of random look-ahead prediction interval is demonstrated in Fig. \ref{fig:Dist_T} for both the linear and polynomial scaling regimes. 

\begin{figure}[ht]
  \centering
  \subfloat[Linear scaling regime: $\gl^p=0.6$.]{\label{fig:Journal_Outage_Random_T_Dist_Linear}\includegraphics[width=0.4\textwidth]{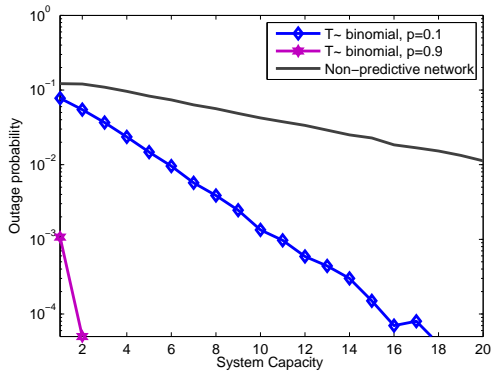}}\\   \subfloat[Polynomial scaling regime: $\gs^p=0.9$.]{\label{fig:Journal_Outage_Random_T_Dist_Sublinear}\includegraphics[width=0.4\textwidth]{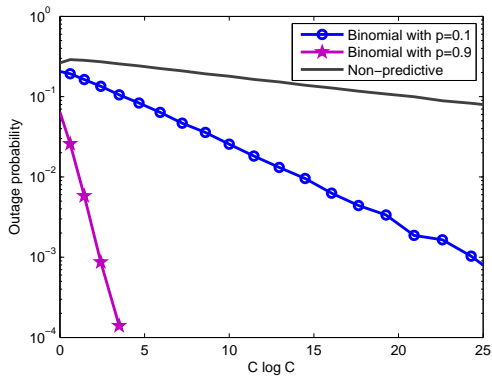}} 
  \caption{Outage probability is significantly improved by proactive networks.}
  \label{fig:Dist_T}
\end{figure} 
The predictive network in each regime is assumed to anticipate requests by a random period which varies between $\Tmin$ and $\Tmax$ where $\Tmin=0$ and $\Tmax=5$. We consider a general binomial distribution with parameter $p$, $0\leq p\leq 1$ to represent the PDF of the look-ahead interval. Hence, the probability that an arriving request at the beginning of time slot $n$ has a deadline at slot $n+T$, $\Tmin\leq T\leq \Tmax$, is given by
\begin{equation}
P(T_{q(n)}=T)=p_{T}=\binom{\Tmax}{T}p^{T}(1-p)^{\Tmax-T}.
\end{equation}  
We consider different values of $p$ in each regime in addition to the non-predictive network scenario. The obtained results for the linear scaling regime are shown in Fig. \ref{fig:Journal_Outage_Random_T_Dist_Linear} where at $p=0.1$, $\dl_{PR}(\gl)\geq\gl p_0-\log(\gl p_0)-1$, and $\dl_{PR}(\gl)=(\Tmax+1)(\gl-1-\log \gl)$ at $p=0.9$. The results of the polynomial scaling regime are shown in Fig. \ref{fig:Journal_Outage_Random_T_Dist_Sublinear}. Although the diversity gain is tantamount to that of the non-predictive network, it is clear from the figure that the outage probability is significantly improved. Here, we want to point out that diversity gain represents the asymptotic decay rate of  the outage probability with the system capacity (or $C\log C$), but it does not capture the relative difference between the outage probability curves themselves. This is why the curves show different trends at small values of $C$. After all, the figure shows that even if $\Tmin=0$ the network achieves a significantly better outage  performance when it follows a proactive resource allocation technique.

Finally, from Figs. \ref{fig:Journal_Outage_Random_T_Dist_Linear}, \ref{fig:Journal_Outage_Random_T_Dist_Sublinear}, we can roughly infer that as $p$ in creases, it is more likely to have arriving requests with larger prediction interval and hence the network gets more degrees of freedom in scheduling such requests in an efficient way that reduces the number of outage events.

\subsection{Two-QoS Network}
\begin{figure}[ht]
  \centering
  \subfloat[Linear scaling regime: $\gl^p=0.6$, $\gl^s=0.1$.]{\label{fig:Journal_S_outage_selfish_primary_linear}\includegraphics[width=0.4\textwidth]{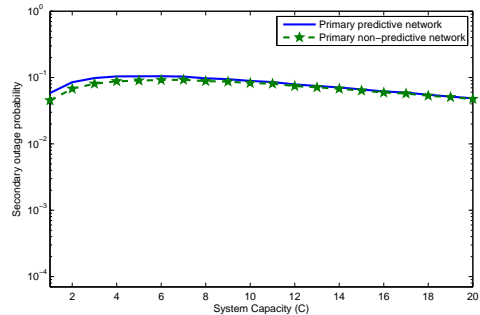}} \\              
  \subfloat[Polynomial scaling regime: $\gs^p=0.75$, $\gs^s=0.05$.]{\label{fig:Journal_S_outage_selfish_primary_sublinear}\includegraphics[width=0.4\textwidth]{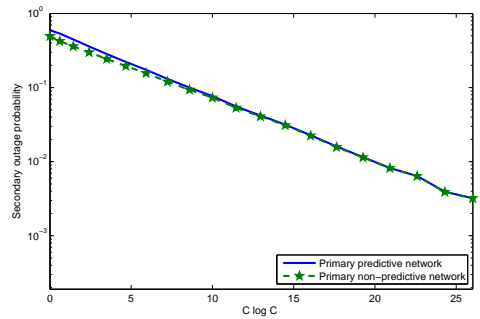}}
  
  \caption{Selfish primary predictive network cannot improve the outage probability of the secondary.}
  \label{fig:Outage_not_improved}
\end{figure}

Fig. \ref{fig:Outage_not_improved} demonstrates the result \eqref{eq:ineq_outage_2} for both the linear scaling and polynomial scaling regimes. The simulation is run assuming $10^3$ time slots and averaged over $10^2$ sample paths. For the selfish predictive primary network, we assume that $T=4$ and the primary requests are served according to EDF. The results of the linear scaling regime are depicted in Fig. \ref{fig:Journal_S_outage_selfish_primary_linear}, whereas that of the polynomial scaling regime are depicted in Fig. \ref{fig:Journal_S_outage_selfish_primary_sublinear}.

\begin{figure}
	\centering
		\includegraphics[width=0.40\textwidth]{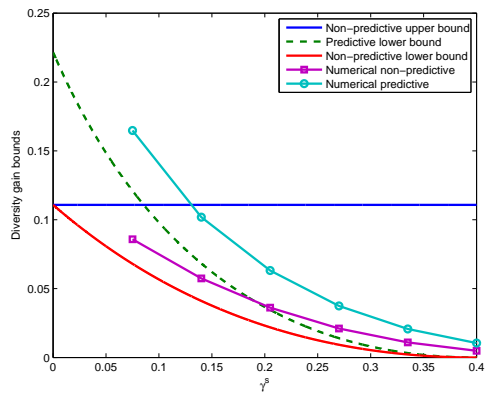}
	\caption{Improvement in the diversity gain of the secondary network under predictive primary with $T=1$ and dynamic capacity assignment. Considered in the figure is the linear scaling regime with $\gl^p=0.6$. The lower bound on $\dl_P(\gl^p,\gl^s)$ is shown in red, and obviously it strictly exceeds the upper bound on $\dl_N(\gl^p,\gl^s)$ determined in Theorem \ref{th:Div_sec_non_pred} plotted in blue.}
	\label{fig:Improvement_in_dPs_T1}
\end{figure}
Figure \ref{fig:Improvement_in_dPs_T1} shows the potential improvement in the diversity gain of the secondary network by efficient use of prediction at the primary side only. Also, simulation results and analytical results are plotted together on the same figure to show the relative differences.

\begin{figure}[ht]
  \centering
  \subfloat[Linear scaling regime: $\gl^p=0.6$, $\gl^s=0.1$.]{\label{fig:Journal_dynamic_capacity_linear}\includegraphics[width=0.4\textwidth]{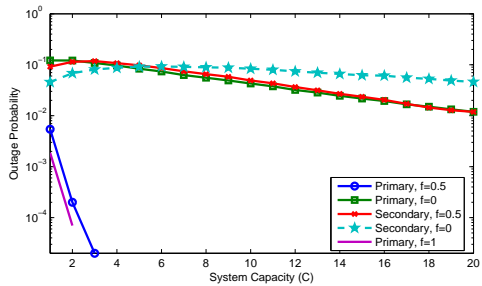}}\\
  \subfloat[Polynomial scaling regime: $\gs^p=0.75$, $\gs^s=0.05$.]{\label{fig:Journal_dynamic_capacity_sublinear}\includegraphics[width=0.4\textwidth]{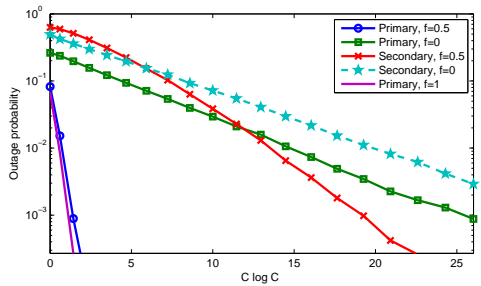}}
  
  \caption{Primary predictive network can tolerate a trivial loss in outage probability at a significant improvement in the secondary outage probability.}
  \label{fig:Dynamic_primary_capacity}
\end{figure}

The performance of the dynamic-primary-capacity scheme, has been evaluated numerically and plotted in Fig. \ref{fig:Dynamic_primary_capacity} for different values of $f$ and under the two scaling regimes, namely, the linear scaling in Fig. \ref{fig:Journal_dynamic_capacity_linear} and the polynomial scaling in Fig. \ref{fig:Journal_dynamic_capacity_sublinear}. The prediction interval is chosen to be $T=4$ and at each slot $n$, the primary network is assumed to serve the $C^p(n)$ primary requests according to EDF policy. For the two schemes, the selfish primary network, at $f=1$, results in the smallest primary outage probability, while at $f=0.5$, the primary outage probability is slightly increased beyond the selfish case, but the secondary outage probability outperforms its counterpart of the non-predictive primary network obtained at $f=0$. It is clear from the figures that at $f=0.5$ the secondary outage probability achieves the primary outage probability of the primary non-predictive network at $f=0$ in the linear scaling regime, and is even better in the polynomial scaling regime. The simulation is for $10^3$ time slots averaged over $10^2$ sample paths.

\subsection{Proactive Multicasting with Symmetric Demands}
\begin{figure}
	\centering
		\includegraphics[width=0.45\textwidth]{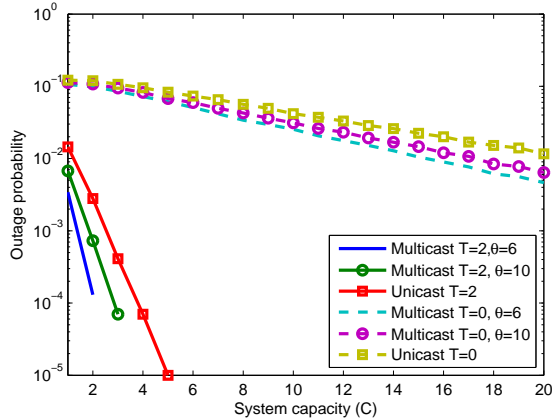}
	\caption{Outage probability versus $C$. In this simulation, $\gu=0.6$.}
	\label{fig:Journal_outage_multi_uni_pred_linear}
\end{figure}

The outage probability of the predictive multicast and unicast networks of the symmetric input traffic is compared numerically to that of non-predictive network and is plotted in Fig. \ref{fig:Journal_outage_multi_uni_pred_linear}. The figure shows the significant enhancement to the outage probability of the multicast network when prediction is employed. Moreover, we can see that the outage probability of the unicast predictive network is better than that of the multicast non-predictive network. The impact of $\thl$ also appears clearly, as it can easily be noticed that as the $\thl$ decreases, the outage performance is enhanced even for the same value of $T$. When $\thl\to\infty$ the multicast curves coincide on the unicast as shown in Section \ref{sec:Multicast}. 

\section{Conclusion and Discussion}
\label{sec:Conc}
We have proposed a novel paradigm for wireless resource allocation which exploits the predictability of user behavior to minimize the spectral resources (e.g., bandwidth) needed to achieve certain QoS metrics. Unlike the traditional reactive resource allocation approach in which the network can only start serving a particular user request upon its initiation, our proposed scheme anticipates future requests. This grants the network more flexibility in scheduling those potential requests over an extended period of time. By adopting the outage (blocking) probability as our QoS metric, we have established the potential of the proposed framework to achieve significant spectral efficiency gains in several interesting scenarios.

More specifically, we have introduced the notion of prediction diversity gain and used it to quantify the gain offered by the proposed resource allocation algorithm under different assumption on the performance of the traffic prediction technique. Moreover, we have shown that, in the cognitive network scenario, prediction at one side only does not only enhance its diversity gain, but it also improves the diversity gain performance of the other user class. On the multicasting front, we have shown that the diversity gain of predictive multicast network scales super-linearly with the prediction window.
Our theoretical claims were supported by numerical results that demonstrate the remarkable gains that can be leveraged from the proposed techniques. 

We believe that this work has only scratched the surface of a very interesting research area which spans several disciplines and could potentially have a significant impact on the design of future wireless networks. In fact, one can immediately identify a multitude of interesting research problems at the intersection of information theory, machine learning, behavioral science, and networking. For example, the analysis have focused on the case of { fixed supply and variable demand}. Clearly, the same approach can be used to { match} demand with supply under more general assumptions on the two processes.


\appendices
\section{Proof of Theorem \ref{th:diversity_reac}}
\label{app:non_pred}
Let $\Lambda_Q(r)$ denote the log moment generating function \cite{Big_Queues} of a Poisson random variable $Q(n), n>0$ with mean $\lambda$, i.e., $$\Lambda_Q(r)=\lambda(e^r-1),\quad r\in\mathbb{R}.$$
For the linear scaling regime, let $\Xl_i,$ $i=1,2,\cdots$ be a sequence of independent and identically distributed (IID) random variables, each with a Poisson distribution with mean $\gl$, and define $$\Sl_C\triangleq\sum_{i=1}^{C}{\Xl_i}.$$ The outage probability, $\PlN$, can then be written as
\begin{equation}
\label{eq:Pl_N_out}
\begin{aligned}
 \PlN&=P(\Ql(n)>C)\\
            &=P\left(\frac{\Sl_C}{C}>1\right)
\end{aligned}            
\end{equation}
Applying Cramer's theorem \cite{Big_Queues} to \eqref{eq:Pl_N_out}, we get
\begin{equation}
\label{eq:Cram_1}
\lim_{C\rightarrow\infty} \frac{1}{C} \log P\left(\frac{\Sl_C}{C}>1\right)=\inf_{r>0}\left\{ \Lambda_X(r)-r\right\},
\end{equation}
where $\Lambda_X(r)=\gl(e^r-1)$. By the convexity of the log moment generating function, we obtain $$\inf_{r>0} \left\{\Lambda_X(r)-r\right\}=1-\gl+\log \gl.$$
Then, it follows that
\begin{equation}
\label{eq:dl_N}
\begin{split}
\dl_N(\gl)&=-\lim_{C\rightarrow\infty} \frac{\log(\Pl)}{C}\\
          &=\gl-1-\log\gl, \quad 0<\gl<1.
\end{split}
\end{equation}

For the polynomial scaling regime, we determine the diversity gain using tight lower and upper bounds.
First, the outage probability is given by
\begin{align}
\label{eq:Ps_N_out}
\PsN&=P(\Qs(n)>C) \\
           &=\sum_{k=C+1}^{\infty}{\frac{(C^{\gs})^k}{k!} e^{- C^{\gs}}} \nonumber \\
           &\geq \frac{(C^{\gs})^{(C+1)}}{(C+1)!}e^{-C^{\gs}} \nonumber.
\end{align} 
Using Stirling's formula to approximate the factorial function, we have $$(C+1)!\doteq\sqrt{2\pi(C+1)}\left(\frac{C+1}{e}\right)^{(C+1)},$$ where $\doteq$ means that the left hand side approaches the right hand side in the limit as $C\rightarrow\infty$. Hence,
\begin{multline*}
\lim_{C\rightarrow\infty}-\frac{\log\PsN}{C\log C}\leq \\ \lim_{C\rightarrow\infty} -\frac{1}{C\log C} \log\left(\frac{e^{-C^{\gs}}}{\sqrt{2\pi(C+1)}}\left(\frac{C^{\gs} e}{C+1}\right)^{C+1}\right)
\end{multline*}
Therefore,
\begin{equation}
\label{eq:ds_N_upper}
\ds_N(\gs)\leq 1-\gs.
\end{equation}
Second, applying tightest Chernoff bound \cite{Big_Queues} on \eqref{eq:Ps_N_out}, we have
\begin{equation}
\label{eq:Chernoff_1}
P(\Qs(n)>C)\leq \inf_{r>0} e^{\Lambda_{\Qs}(r)-rC}
\end{equation}
where $\Lambda_{\Qs}(r)=C^{\gs}(e^r-1)$. And since $\Lambda_{\Qs}(r)-r$ is convex in $r$, by simple differentiation, we get
\begin{equation}
\label{eq:Chernoff_2}
\PsN\leq  e^{C-C^{\gs}-(1-\gs)C\log C}.
\end{equation}
Now, taking the logarithm of both sides of \eqref{eq:Chernoff_2}, dividing by $-C\log C$, and letting $C\rightarrow\infty$, it follows that
\begin{equation}
\label{eq:ds_N_lower}
\ds_N(\gs)\geq 1-\gs.
\end{equation}
By \eqref{eq:ds_N_upper}, \eqref{eq:ds_N_lower}, $$1-\gs\leq \ds_N(\gs) \leq 1-\gs,$$ then
\begin{equation}
\label{eq:ds_N}
\ds_N(\gs)=1-\gs, \quad 0<\gs<1.
\end{equation}

\section{Proof of Lemma \ref{lem:Det_T}}
\label{app:Det_T}
For $\mathcal{U}_D$, we need to show that the outage occurring at time slot $n$ implies $\sum_{i=0}^{T}{Q(n-T-i)}>C(T+1)$. To see this, assume there is an outage at slot $n$. Since in our scenario EDF reduces to FCFS, then: 1) the outage at slot $n$ occurs only on the arrivals of slot $n-T$ and 2) during the interval of slots $n-T, n-T+1, \cdots, n$, the system does not serve any of the arriving requests at slots beyond $n-T$. Let $N(m), m>0$ denote the number of requests in the system at the beginning of slot $m$, then having an outage at slot $n$ implies $N(n-T)>C(T+1)$. And since at any slot $m>0$, there are no requests in the system arriving at slots prior to $m-T$, it follows that $\sum_{i=0}^{T} Q(n-T-i)\geq N(n-T)>C(T+1)$.

For $\mathcal{L}_D$, we need to show that $Q(n-T)>C(T+1)$ implies an outage at slot $n$. This is straightforward as the arrivals at slot $n-T$ can not remain in the system at any slot beyond $n$, furthermore, since $Q(n-T)>C(T+1)$, the capacity of the system at the slot of arrival in addition to the next $T$ slots is not sufficient to serve the $Q(n-T)$ requests, hence the system encounters an outage at slot $n$.

\section{Proof of Theorem \ref{th:diversity_proac_d}}
\label{app:div_pred_d}
For the linear scaling regime, we have from Lemma \ref{lem:Det_T}, $\Pl_{PD}\leq P({\overline{\mathcal{U}_D}})$, hence,
\begin{equation}
\PlPd\leq P\left(\sum_{i=0}^{T} \Ql(n-T-i)>C(T+1)\right).
\end{equation}
Using the same definition of the sequence of IID random variables $X_i, i>0$ as in the proof of Theorem \ref{th:diversity_reac}, we have $\Sl_{C(T+1)}=\sum_{i=1}^{C(T+1)}{X_i}$ and 
\begin{equation}
P\left(\sum_{i=0}^{T} {\Ql(n-T-i)}>C(T+1)\right)=P\left(\frac{\Sl_{C(T+1)}}{C(T+1)}>1\right).
\end{equation}
Using Cramer's theorem,
\begin{equation}
\lim_{C\rightarrow\infty} -\frac{\log P(\overline{\mathcal{U}_D})}{C(T+1)}=\gl-1-\log \gl.
\end{equation}
Since $P^*_{PD}(\overline{\out})\leq\PlPd\leq P(\overline{\mathcal{U}_D})$, we have
\begin{equation}
\begin{split}
\lim_{C\rightarrow\infty}-\frac{\log P^*_{PD}(\overline{\out})}{C}&\geq\lim_{C\rightarrow\infty} -\frac{\log P(\overline{\mathcal{U}_D})}{C}\\
                                                        &=(T+1)(\gl-1-\log \gl),
\end{split}
\end{equation}
for which \eqref{eq:dl_Pd_lower} follows.

For the polynomial scaling regime, first we use the upper bound $\PsPd\leq P(\widetilde{\mathcal{U}_D})$ to establish a lower bound on the optimal diversity gain $\ds_{PD}(\gs)$ as follows.
Using Chernoff bound on $P(\widetilde{\mathcal{U}_D})$,
\begin{equation}
\begin{split}
\PsPd&\leq P\left(\sum_{i=0}^{T}\Qs(n-T-i)>C(T+1)\right)\\
              &\leq \inf_{r>0} e^{(T+1)\Lambda_{\Qs}(r)-C(T+1)r},
\end{split}
\end{equation}
where $\Lambda_{\Qs}(r)=C^{\gs}(e^{r}-1)$. Then, using differentiation,
\begin{equation}
\PsPd\leq e^{(T+1)(C-C^{\gs})-(T+1)(1-\gs)C\log C}.
\end{equation}
And since $P^*_{PD}(\widetilde{\out})\leq \PsPd$, we get
\begin{equation}
\label{eq:ds_Pd_lower}
\ds_{PD}(\gs)\geq (1+T)(1-\gs).
\end{equation}
Second, we use the lower bound $\PsPd\geq P(\widetilde{\mathcal{L}_D})$ to establish an upper bound on $\ds_{PD}(\gs)$.
\begin{equation*}
\begin{split}
P(\widetilde{\mathcal{L}_D})&=P(\Qs(n-T)>C(T+1))\\
                  &=\sum_{k=C(T+1)+1}^{\infty}{\frac{(C^{\gs})^k}{k!} e^{- C^{\gs}}}\\
                  &\geq \frac{(C^{\gs})^{(C(T+1)+1)}}{(C(T+1)+1)!}e^{-C^{\gs}}\\
                  &\doteq \frac{e^{-C^{\gs}}}{\sqrt{2\pi(C(T+1)+1)}}\left(\frac{C^{\gs} e}{C(T+1)+1}\right)^{C(T+1)+1}
\end{split}
\end{equation*}
And since $P^*_{PD}(\widetilde\out)\leq\PsPd$, we obtain
\begin{equation}
\label{eq:ds_Pd_upper}
\ds_{PD}(\gs)\leq (1+T)(1-\gs).
\end{equation}
By \eqref{eq:ds_Pd_lower}, \eqref{eq:ds_Pd_upper}, it follows that $$\ds_{PD}(\gs)=(1+T)(1-\gs),\quad 0<\gs<1.$$

\section{Proof of Lemma \ref{lem:Random_T}}
\label{app:Random_T}
First, we show that $\Ur$ is a necessary condition for the outage event, that is, if an outage occurs at slot $n$, then $\Ur=\I\cup\J$ occurs. Suppose there is an outage at slot $n$. This outage occurs on the arrivals, $Q_{k}(n-k), \quad k=\Tmin,\cdots,\Tmax$, hence, $\sum_{i=0}^{\Tmin}N_{i}(n-\Tmin)>C(\Tmin+1)$, i.e., in the interval $n-\Tmin, \cdots, n$ the system is serving requests with deadlines not exceeding $n$.

Event $\I$ represents the case when at slot $n-\Tmax$, the number of requests in the system in addition to the requests that will arrive with deadlines not larger than $n$ is larger than $C(\Tmax+1)$, i.e., larger than the maximum number of requests that the system can serve in the subsequent $\Tmax+1$ slots (Fig. \ref{fig:Lem2_a} shows the requests considered in event $\I$ as blue circles for $\Tmin=1$, $\Tmax=3$.). However, event $\I$ alot is not a necessary condition for an outage as, for instance, we may have $Q_{\Tmin}(n-\Tmin)>C(\Tmin+1)$ but $\sum_{j=0}^{\Tmax}\sum_{i=\Tmin}^{\Tmax}Q_i(n-i-j)<C(\Tmax+1)$.

Now, suppose that $I$ did not occur because of the outage at slot $n$, then there exists at least one slot $n-l$, $\Tmin<l\leq\Tmax$ such that $\sum_{i=0}^{l}N_i(n-l)\leq C$ (Otherwise, the system will be serving requests with deadline of at most $n$ in slots $n-\Tmax, \cdots, n-\Tmin$ which implies $n\in\I$.). In other words, at slot $l$, the system will be empty of all requests that have deadlines not beyond slot $n$. Let $$l^*=\min\left\{l:\sum_{i=0}^{l}N_i(n-l)\leq C, \quad \Tmin<l\leq\Tmax\right\},$$ then $\sum_{j=\Tmin}^{l^*-1}\sum_{i=\Tmin}^{j}Q_i(n-j)>Cl^*$, hence $J$ occurs. That is, all of the arriving requests in slots $n-l^*+1, \cdots, n-\Tmin$ with deadlines not beyond $n$ are more than $Cl^*$ (Fig. \ref{fig:Lem2_b} shows the case when event $\I$ is not occurring while $l^*=3$.).

\begin{figure}[ht]
  \centering
  \subfloat[Blue circles represent an upper bound on the requests that \textbf{must} be served by slot $n$ in the interval of slots $n-\Tmax, \cdots, n$. Red circles represent requests that are no longer in the system at slot $n-\Tmax$ whereas white circles represent requests with deadline larger than $n$.]{\label{fig:Lem2_a}\includegraphics[width=0.4\textwidth]{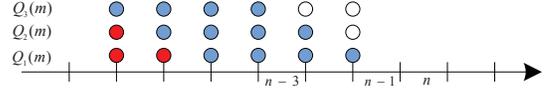}}\\                
  \subfloat[Here, event $\I$ is not satisfied. At slot $n-3$ the system has managed to serve all requests with deadlines not exceeding $n$. However, $l^*=3$, meaning that all of the next arrivals with deadlines not exceeding $n$ will consume the whole system capacity till slot $n$ inclusive.]{\label{fig:Lem2_b}\includegraphics[width=0.4\textwidth]{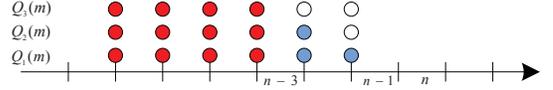}}
  
  \caption{An outage occurs at slot $n$ where $\Tmin=1$, $\Tmax=3$. At the beginning of any time slot, arriving requests with the same deadline are represented by a circle.}
  \label{fig:Example_Random_T}
\end{figure}

Second, we show that $\Lr$ is a sufficient condition on the outage event. The proof is straightforward as for every $k$, $\Tmin\leq k\leq\Tmax$, the event that $\sum_{i=\Tmin}^{k}Q_i(n-i)>C(k+1)$ means the number of requests that must be served in the interval $n-k, \cdots, n$ is larger than $C(k+1)$ which is sufficient to cause an outage at slot $n$. Then, taking the union over all $k\in\{\Tmin, \cdots, \Tmax\}$ is also a sufficient condition for an outage at slot $n$.  

\section{Proof of Theorem \ref{th:Random_T}}
\label{app:R_T}
\begin{equation*}
\begin{split}
P_{PR}(\out) &\leq P(\Ur)\\
             &\leq P\left(\sum_{j=0}^{\Tmax}\sum_{i=\Tmin}^{\Tmax}Q_i(n-j-i)>C(\Tmax+1)\right)\\
             &\quad +\sum_{k=\Tmin}^{\Tmax-1}P\left(\sum_{j=\Tmin}^{k}\sum_{i=\Tmin}^{j}Q_i(n-j)>C(k+1)\right).\\
             &\leq \inf_{r_I>0}e^{\Lambda_{Q_I}(r_I)-r_IC(\Tmax+1)}\\
             &\quad +\sum_{k=\Tmin}^{\Tmax-1}\inf_{r_k>0}e^{\Lambda_{Q_k}(r_k)-r_kC(k+1)}
\end{split}
\end{equation*}
where $\Lambda_{Q_I}(r_I)=\lambda (\Tmax+1)(e^{r_I}-1)$ and $\Lambda_{Q_k}(r_k)=\lambda\sum_{i=0}^{k-\Tmin}F_{k-i}$, $\Tmin\leq k\leq\Tmax-1$.

For the linear scaling regime,
\begin{equation*}
\begin{split}
\PlPr&\leq e^{(1-\gl)C(\Tmax+1)+C(\Tmax+1)\log\gl}\\
              &\quad +\sum_{k=\Tmin}^{\Tmax-1} e^{C(k+1)\left[1-\frac{\gl \sum_{i=0}^{k-\Tmin}F_{k-i}}{k+1}+\log\frac{\gl \sum_{i=0}^{k-\Tmin}F_{k-i}}{k+1}\right]}.
\end{split}
\end{equation*}
Let
\begin{multline*}
\vl(C)\triangleq\max_{\Tmin\leq k\leq \Tmax-1}\left\{C(k+1)\left[1-\frac{\gl \sum_{i=0}^{k-\Tmin}F_{k-i}}{k+1} \right.\right. \\
\left.\left.+\log\frac{\gl \sum_{i=0}^{k-\Tmin}F_{k-i}}{k+1}\right]\right\}
\end{multline*}
and
\begin{equation*}
\ml(C)=\max\left\{C(\Tmax+1)(1-\gl+\log\gl), \vl(C)\right\},
\end{equation*}
then
\begin{align}
\dl_{PR}(\gl)&\geq\lim_{C\rightarrow\infty}-\frac{1}{C}\log e^{\ml(C)} \nonumber \\
             &=\lim_{C\rightarrow\infty} -\frac{\ml(C)}{C} \nonumber \\
             &=\min\{(\Tmax+1)(\gl-1-\log \gl), \vl_*\}
\end{align}
which proves \eqref{eq:dl_Pr_Lower}.

For the polynomial scaling regime,
\begin{equation*}
\begin{split}
\PsPr&\leq e^{(\Tmax+1)(C-C^{\gs}-C\log C^{1-\gs})}\\
              &\quad+\sum_{k=\Tmin}^{\Tmax-1}e^{C(k+1)\left[1-\frac{\sum_{i=0}^{k-\Tmin}{F_{k-i}}}{C^{1-\gs}(k+1)}+\log\frac{\sum_{i=0}^{k-\Tmin}F_{k-i}}{C^{1-\gs}(k+1)}\right]}.
\end{split}
\end{equation*}
Let
\begin{multline}
\label{eq:vs}
\vs(C)=\max_{\Tmin\leq k\leq \Tmax-1}\left\{C(k+1)\left[1-\frac{\sum_{i=0}^{k-\Tmin}{F_{k-i}}}{C^{1-\gs}(k+1)} \right.\right. \\ \left.\left.+\log\frac{\sum_{i=0}^{k-\Tmin}F_{k-i}}{C^{1-\gs}(k+1)}\right]\right\}
\end{multline}
and
\begin{equation*}
\ms(C)\triangleq \lim_{C\to\infty}-\frac{\ms(C)}{C\log C},
\end{equation*}
for large values of $C$, the terms in the $\max\{.\}$ of \eqref{eq:vs} are decreasing in $k$, hence
\begin{equation}
\label{eq:ds_Pr_Lower}
\ds_{PR}(\gs)\geq (\Tmin+1)(1-\gs).
\end{equation}

Then, we use the event $\Lr$ with the polynomial scaling as follows.
\begin{equation*}
\begin{split}
\PsPr&\geq P(\widetilde{\Lr})\\
              &\geq \max_{\Tmin\leq k\leq \Tmax}\left\{P\left(\sum_{i=\Tmin}^{k}\Qs_i(n-i)>C(k+1)\right)\right\}\\
              &\geq \max_{\Tmin\leq k\leq \Tmax}\frac{\left(F_kC^{\gs}\right)^{C(k+1)+1}}{C(k+1)+1!}e^{-F_kC^{\gs}}\\
              &\doteq \max_{\Tmin\leq k\leq \Tmax} \left(\frac{F_kC^{\gs}e}{C(k+1)+1}\right)^{C(k+1)+1}\\
              &\quad\times\frac{e^{-F_k C^{\gs}}}{\sqrt{2\pi (C(k+1)+1)}}\\
              &= \left(\frac{F_k C^{\gs}e}{C(\Tmin+1)+1}\right)^{C(\Tmin+1)+1}\\
              &\quad\times\frac{e^{-p_{\Tmin}C^{\gs}}}{\sqrt{2\pi (C(\Tmin+1)+1)}}.
\end{split}
\end{equation*}
Hence,
\begin{equation}
\label{eq:ds_Pr_Upper}
\ds_{PR}(\gs)\leq (\Tmin+1)(1-\gs).
\end{equation}
From \eqref{eq:ds_Pr_Lower} and \eqref{eq:ds_Pr_Upper}, result \eqref{eq:ds_Pr} follows.

\section{Proof of Theorem \ref{th:Div_sec_non_pred}}
\label{app:Div_sec_non_pred}
Let the outage probability of the secondary user while the primary network is non-predictive be denoted by $P_N(\out^s)$, then
\begin{equation}
\label{eq:P_N^s}
P_N(\out^s)=P(Q^p(n)+Q^s(n)>C, Q^s(n)>0).
\end{equation}
Since $Q^p(n)+Q^s(n)$ and $Q^s(n)$ are two dependent random variables, we use upper and lower bounds on $P_N^s(\mathcal{O})$ to characterize $d_N^s(\gamma^p,\gamma^s)$ as follows.
\begin{align}
P_N(\out^s)&= P(Q^p(n)+Q^s(n)>C|Q^s(n)>0)P(Q^s(n)>0) \nonumber \\
           &\overset{(a)}{\geq} P(Q^p(n)>C|Q^s(n)>0)P(Q^s(n)>0) \nonumber \\
\label{eq:P_N^s_lower}           
           &\overset{(b)}{=} P(Q^p(n)>C)P(Q^s(n)>0),
\end{align}
where (a) follows from the fact that $Q^s(n)\geq 0$ and (b) follows as $Q^p(n)$ and $Q^s(n)$ are independent. 
Moreover, since $P(\mathcal{A},\mathcal{B})\leq P(\mathcal{A})$, then, from \eqref{eq:P_N^s}, we can write
\begin{equation}
\label{eq:P_N^s_upper}
P_N(\out^s)\leq P(Q^p(n)+Q^s(n)>C).
\end{equation}
For the linear scaling regime, we have $\bar{\lambda}^p=\gl^pC$ and $\bar{\lambda}^s=\gl^sC$. From \eqref{eq:a1}, \eqref{eq:a2} we obtain $0<\gl^s<\gl^p<1$ and $\gl^s+\gl^p<1$.
From \eqref{eq:P_N^s_lower},
\begin{equation*}
\begin{split}
\PlNs &\geq  P(\Ql^p(n)>C) P(\Ql^s(n)>0)\\
             &=     P(\Ql^p(n)>C) \left(1-e^{-\gl^sC}\right).
\end{split}
\end{equation*}
Hence
\begin{equation}
\begin{split}
\dl_N^s(\gl^p,\gl^s)&\leq \lim_{C\rightarrow\infty} \frac{-\log P(\Ql^p(n)>0)}{C}-\frac{\log\left(1-e^{-\gl^sC}\right)}{C}\\
                    &\overset{(c)}{=} \gl^p-1-\log(\gl^p),
\end{split} 
\end{equation}
where (c) follows by Cramer's theorem. This proves \eqref{eq:dl_N^s_upper}.
Since $\Ql^p(n)$, $\Ql^s(n)$ are independent Poisson random variables, then $\Ql^p(n)+\Ql^s(n)$ is a Poisson process with rate $(\gl^p+\gl^s)C$.  Applying Cramer's theorem to \eqref{eq:P_N^s_upper}, we obtain
\begin{equation*}
\dl_N^s(\gl^p,\gl^s) \geq (\gl^p+\gl^s)-1-\log(\gl^p+\gl^s)
\end{equation*}
which proves \eqref{eq:dl_N^s_lower}.

For the polynomial scaling regime, $\tilde{\lambda}^p=C^{\gs^p}$, $\tilde{\lambda}^s=C^{\gs^s}$. From \eqref{eq:a1}, \eqref{eq:a2}, we get $0<\gs^s<\gs^p<1$.
From \eqref{eq:P_N^s_upper},
\begin{equation*}
\begin{split}
\PsNs&\geq P(\Qs^p(n)>C)P(\Qs^s(n)>0)\\
             &= P(\Qs^p(n)>C)(1-e^{-C^{\gs^s}})\\
             &\geq \frac{C^{\gs^p{(C+1)}}}{(C+1)!}e^{-C^{\gs^p}}\left(1-e^{-C^{\gs^s}}\right)\\
             &\doteq \left(\frac{C^{\gs^p}e}{C+1}\right)^{C+1} \frac{e^{-C^{\gs^p}}}{\sqrt{(2\pi (C+1))}}\left(1-e^{-C^{\gs^s}}\right).
\end{split}
\end{equation*}
Hence,
\begin{equation}
\label{eq:ds_N^s_upper}
\begin{split}
\ds_N^s(\gs^p,\gs^s)&\leq \lim_{C\rightarrow\infty} \frac{-\log P(\Qs^p(n)>C)}{C\log C}-\frac{\left(1-e^{-C^{\gs^s}}\right)}{C\log C}\\
                    &\leq 1-\gs^p.
\end{split}
\end{equation}
From \eqref{eq:P_N^s_upper}, we obtain, using tightest Chernoff bound,
\begin{equation}
\PsNs\leq \inf_{r>0} e^{\Lambda_{\Qs^s+\Qs^p}(r)-rC},
\end{equation}
where $\Lambda_{\Qs^p+\Qs^r}(r)=(C^{\gs^p}+C^{\gs^s})(e^r-1)$. Then it follows that,
\begin{equation}
\label{eq:ds_N^s_lower}
\begin{split}
\ds_N^s(\gs^p,\gs^s)&\geq 1-\max\{\gs^p,\gs^s\}\\
                    &=1-\gs^p.
\end{split}
\end{equation}
From \eqref{eq:ds_N^s_upper} and \eqref{eq:ds_N^s_lower}, the result \eqref{eq:ds_N^s} follows.

\section {Proof of Theorem \ref{th:dynamic}}
\label{app:dynamic}
Let the outage probability of the primary network under the dynamic scheduling policy be denoted by $P_P(\out^p)$. To upper bound this outage probability, it suffices to show that $f\in[0.5,1]$ implies $P_P(\out^p)\leq P(\mathcal{U}_D)$, where $\mathcal{U}_D$ is as defined in Lemma \ref{lem:Det_T}. So, suppose that there is an outage at slot $n$, hence, according to the dynamic policy, $C^p(n)=C$ as $N_0^p(n)>C$. Moreover, that outage is occurring on $Q^p(n-T)$.

Now, at time slot $n-1$, assume towards contradiction that $C^p(n-1)<C$, then $f N_1(n-1)<C$. This must lead to $N_0(n)\leq (1-f) N_1(n-1)<C$ as $1-f\leq f$, $f\in[0.5,1]$, which is a contradiction. Therefore, $C^p(n-1)=C$.

Since the EDF nature of the dynamic policy implies that the network resources are only dedicated to serve primary requests that arrived prior to slot $n-T+1$, then $C^p(n-1)$ and $C^p(n)$ represent the served requests that arrived at slots $n-T-1$ and $n-T$.
But, $C^p(k)\leq \min\{C, f(C^p(n-1)+C^p(n))\}$, $k=n-T,\cdots, n$. Hence, $C^p(k)=C$ for all $k=n-T,\cdots, n$ as $f\in[0.5,1]$.

Therefore, an outage at slot $n$ implies $\sum_{i=0}^{T}Q^p(n-i-T)>C(T+1)$, and consequently, we obtain the lower bounds on $\dl_P^p(\gl^p)$ and $\ds_P^s(\gs^p)$ in the same manner as in Theorem \ref{th:diversity_proac_d}.

Also, it is straightforward to see that the event $\mathcal{L}_D$ of Lemma 1 satisfies $P(\mathcal{L}_D)\leq P_P(\out^p)$. So the diversity gain of the polynomial scaling regime is fully determined.

\section{Proof of Theorem \ref{th:div_sec_dyn}}
\label{app:div_sec_dyn}
We will show the result for the linear scaling regime while its polynomial scaling regime counterpart is obtained through the same approach by taking into account the difference in the diversity gain definitions.

From \eqref{eq:P_P(Os)1} and \eqref{eq:T_1_cap}, we can upper bound $P_P(\out^s)$ by
\begin{align*}
P_P(\out^s) & \leq P(Q^s(n)+N_0^p(n)+0.5 Q^p(n)>C, \\
            & \quad C^p(n-1)<C)\\
            & \quad +P(Q^s(n)+N_0^p(n)+0.5 Q^p(n)>C,\\
            & \quad C^p(n-1)=C).
\end{align*}
But $C^p(n-1)<C$ implies $N_0^p(n)=0.5 Q^p(n-1)$ and hence the joint event $Q^s(n)+N_0^p(n)+0.5 Q^p(n)>C,C^p(n-1)<C$ implies $Q^s(n)+0.5 Q^p(n-1)+0.5 Q^p(n)>C$. Therefore, 
\begin{equation}
\label{eq:P_P(Os)2}
\begin{split}
P_P(\out^s) & \leq P(Q^s(n)+0.5 Q^p(n-1)+0.5 Q^p(n)>C)\\
            & \quad +P(C^p(n-1)=C).
\end{split}            
\end{equation}
Now, we show that the decay rate of the second term on the right hand side of \eqref{eq:P_P(Os)2} with $C$ is larger than the first. We start with the second term $P(C^p(n-1)=C)$ which can be upper bounded by
\small
\begin{align}
P(C^p(n-1)=C)&\leq P(N_0^p(\eta)+0.5 Q^p(\eta)>C, \nonumber \\
           & \quad C^p(\eta-1)<C\text{ for some }\eta\leq n-1) \nonumber\\
           &\quad + P(N_0^p(m)+0.5 Q^p(m)>C,\nonumber\\
           &\quad C^p(m-1)=C \text{ for all }m \leq n-1)\nonumber\\
           & \leq P(0.5 Q^p(\eta-1)+0.5 Q^p(\eta)>C)\nonumber\\
           \label{eq:all_C}
           &\quad + P(C^p(m)=C, \text{ for all } m\leq n-1).
\end{align}
\normalsize
Fix $0\leq M \leq n-1$. The last term on the right hand side of \eqref{eq:all_C} satisfies
\small
\begin{equation*}
\begin{split}
P(C^p(m)=C,\text{ for all }m\leq n-1)\leq & P(C^p(1)=\\ 
                                          & \cdots = C^p(M)=C),
\end{split}
\end{equation*}
\normalsize
where
\small
\begin{flalign*}
P(C^p(1)=\cdots=C^p(M)=C)&\leq \quad\quad \quad &
\end{flalign*}
\begin{align*}
\quad\quad & P(C^p(1)=\cdots=C^p(M)=C,\text{ No outages in }1,\cdots, M)\\
           & + \sum_{l=1}^{M} P(C^p(1)=\cdots=C^p(M)=C, l \text{ outages in }1,\cdots, M)
\end{align*}
\normalsize
implying
\small
\begin{flalign*}                         
P(C^p(1)=\cdots=C^p(M)=C)&\leq \quad\quad\quad & 
\end{flalign*}
\begin{align*}
\quad\quad & P(C^p(1)=\cdots=C^p(M)=C, \text{ No outages in }1,\cdots, M)\\
           & + (2^M-1)P_P^p(\out^p).
\end{align*}
\normalsize
Since $M$ is constant, the term $(2^M-1) P_P^p(\out^p)$ decays with the system capacity as $d_P^p(\gamma^p)$.
The joint event $C^p(1)=\cdots=C^p(M)=C$ and no outage in $1,\cdots, M$ implies
\begin{equation*}
\begin{split}
N_0^p(M)&=N_0^p(1)-(M-1)C+\sum_{i=1}^{M-1}Q^p(i)\\
        &\leq -(M-1)C+\sum_{i=0}^{M-1}Q^p(i).
\end{split}
\end{equation*}
and hence,
\small
\begin{align*}
 & P(C^p(1)=\cdots=C^p(M)=C, \text{ No outage in }1,\cdots,M) \\
\leq &  P\left(-(M-1)C+\sum_{i=0}^{M-1}Q^p(i)+0.5Q^p(M)>C\right)    \\
\leq & P\left(\sum_{i=0}^{M}Q^p(i)>MC\right)                      \\
\leq &  \inf_{r>0}\left\{e^{\Lambda(r)-rMC}\right\},
\end{align*}
\normalsize
where, for the linear scaling regime,  $$\overline{\Lambda}(r)=(M+1)\gl^pC(e^r-1).$$ Hence, 
\small
\begin{multline}
\label{eq:div_no_outage}
\lim_{C\to \infty}-\frac{1}{C}\log P\Bigl(\overline{C}^p(1)=\cdots=\overline{C}^p(M)=C,\\ \text{ No outage in }1,\cdots,M\Bigr)\geq \\
(M+1)\gl^p-M+M\log\left(\frac{M}{(M+1)\gl^p}\right)
\end{multline}
\normalsize
with the right hand side of \eqref{eq:div_no_outage} monotonically increasing in $M$ as long as $\frac{M}{M+1}>\gl^p$. Then, $M$ can be chosen sufficiently large\footnote{The system is assumed to operate in the steady state, i.e., $n\gg 1$.} so that
\small
\begin{multline*}
\lim_{C\to \infty}-\frac{1}{C}\log P\left(\overline{C}^p(m)=C\text{ for all }m\leq n-1\right)\geq\dl_P^p(\gl^p)\\
                                                                                             =2(\gl^p-1-\log \gl^p).
\end{multline*}
\normalsize

Also, the first term on the right hand side of \eqref{eq:all_C} can be written as
\small
\begin{equation*}
\begin{split}
P(0.5Q^p(\eta-1)+0.5Q^p(\eta)>C)&=P(Q^p(\eta-1)+Q^p(\eta)>2C)\\
                                & \geq P_P^p(\out^p),
\end{split}
\end{equation*}
\normalsize
where $T=1$.
Hence,
\begin{equation}
\begin{split}
\lim_{C\to \infty}-\frac{\log P\left(\overline{C}^p(n-1)=C\right)}{C}&\geq \dl_P^p(\gl^p)\\
                                                                     &=2(\gl^p-1-\log \gl^p).
\end{split}
\end{equation}

Now, comparing the two terms $P(Q^s(n)+0.5 Q^p(n-1)+0.5Q^p(n)>C)$ in \eqref{eq:P_P(Os)2} and $P(0.5Q^p(\eta-1)+0.5Q^p(\eta)>C)$ in \eqref{eq:all_C}, we have by the stationarity of $Q^p(n),n>0$ and the non-negativity of $Q^s(n), n>0$,
\small
\begin{multline*}
P(Q^s(n)+0.5 Q^p(n-1)+0.5Q^p(n)>C)\geq\\ P(0.5Q^p(\eta-1)+0.5Q^p(\eta)>C).
\end{multline*}
\normalsize 
This implies that the asymptotic decay rate of $\log P_P(\out^s)$ with $C$ is lower bounded by the decay rate of $P(Q^s(n)+0.5 Q^p(n-1)+0.5Q^p(n)>C)$ with $C$.

Now, we can use Chernoff bound to lower bound $\dl_P^s(\gl^p,\gl^s)$ as follows
\begin{equation}
\label{eq:dominating_term}
P(\Ql^s(n)+0.5 \Ql^p(n-1)+0.5\Ql^p(n)>C)\leq \inf_{r>0}\left\{e^{\overline{\Lambda}_{tot}(r)-rC}\right\},
\end{equation}
where $$\overline{\Lambda}_{tot}(r)=\gl^sC(e^r-1)+2\gl^pC(e^{0.5r}-1).$$
By differentiation, the optimal value of $r$, denoted $r^*$, satisfies
\begin{equation*}
\gl^s {e^r}^*+{\gl^p}^{0.5 r^*}-1=0.
\end{equation*}
Let $\yl\triangleq e^{0.5 r^*}$, we obtain
\begin{equation*}
\yl=-\frac{{\gl^p}}{2\gl^s}+\frac{\sqrt{4\gl^s+{\gl^p}^2}}{2\gl^s}
\end{equation*}
and $$r^*=2\log \yl.$$
Substituting with $r^*$ in \eqref{eq:dominating_term}, taking $-\log$ of both sides, dividing by $C$ and sending $C\to\infty$, the diversity gain of the secondary network in the linear scaling regime satisfies
\begin{equation*}
\dl_P^s(\gl^p,\gl^s)\geq -\gl^s(\yl^2-1)-2\gl^p(\yl-1)+2\log(\yl).
\end{equation*}

\section{Proof of Theorem \ref{th:Div_multi_rea}}
\label{app:Div_multi_rea}
\begin{equation*}
\begin{split}
P_N(\outm)&=P\left(S^m(n)>C\right)\\
          &=P\left(\frac{S^m(n)}{\thl C}>\frac{1}{\thl}\right)\\
          &=P\left(\frac{\sum_{l=1}^{\thl C}X^{[l]}}{\thl C}>\frac{1}{\thl}\right).
\end{split}
\end{equation*}
Applying Cramer's Theorem \cite{Big_Queues},
\begin{equation}
\label{eq:d_N(symm_multi)}
\dl_N(\gm,\thl)=-\inf_{r>0}\{ \thl \Lambda_{X^{[l]}}(r)-\thl r\},
\end{equation}
but
\begin{equation*}
\begin{split}
\Lambda_{X^{[l]}}(r)&=\log(1-A^m+A^m e^{r})\\
                    &=\log\left(e^{-\frac{\gm}{\thl}}+\left(1-e^{-\frac{\gm}{\thl}}\right)e^r\right),
\end{split}
\end{equation*}
Then, $$r^*=\log\left(\frac{e^{-\frac{\gm}{\thl}}}{(\thl-1)\left(1-e^{-\frac{\gm}{\thl}}\right)}\right).$$
The conditions $0<\gl<1$, $\thl>1$ ensure that $r^*>0$. Substituting with $r^*$ in \eqref{eq:d_N(symm_multi)}, we obtain \eqref{eq:Div_multi_rea_lin}. 

\section{Proof of Theorem \ref{th:div_multi_pred}}
\label{app:div_multi_pred}
Under EDF scheduling, an outage occurs in slot $n\gg 1$ if and only if $N^m(n-T)>C(T+1)$, where $N^m(n-T)$ is the number of distinct multicast data sources targeted by existing requests in the system at slot $n-T$. Hence
\begin{equation*}
P^*_P(\outm)\leq P_P(\outm)=P(N^m(n-T)>C(T+1)). 
\end{equation*}

Let $Z^m_T(n-T)$ be the number of distinct data sources that were requested in the window of slots $[n-2T, \cdots, n-T]$, then according to EDF, 
\begin{equation}
\label{eq:Aux_ineq}
N^m(n-T)\leq Z^m_T(n-T).
\end{equation}
Therefore $P(N^m(n-T)\leq C(T+1))\leq P(Z^m_T(n-T)>C(T+1))$.

Since each data source is requested independently of the others at each slot and from slot to another, then the probability that a data source is requested at least once in a window of $T+1$ slots, denoted $\xm$, is equal to
\begin{equation*}
\begin{split}
\xm&=1-(1-A^m)^{T+1}\\
   &=1-\exp\left(-\frac{(T+1)\gm}{\thl}\right),
\end{split}
\end{equation*}
hence \footnotesize{$$P(Z^m_T(n-T)=k)=\begin{cases} \binom{\thl C}{k}{\xm}^{k}(1-\xm)^{\thl C-k},&k=0,\cdots,\thl C\\ 0, & \text{otherwise.}\end{cases}$$}
\normalsize

Now we can upper-bound $P^*_P(\outm)$ using Chernoff bound as
\begin{equation*}
\begin{split}
P^*_P(\outm)&\leq P_P(\outm)\\
            &\leq P(Z^m_T(n-T)>C(T+1))\\
            &\leq \inf_{r>0} \{e^{\Lambda_Z(r)-rC(T+1)}\},
\end{split}
\end{equation*}
where $\Lambda_Z(r)=\thl C\log\left(1-\xm+\xm e^r\right).$
Solving for $r^*>0$ that minimizes $e^{\Lambda_Z(r)-rC(T+1)}$, we obtain $$r^*=\log\left(\frac{(1-\xm)(T+1)}{\xm(\thl-(T+1))}\right).$$

Now, taking $-\log P^*_P(\gm,\thl)$, dividing by $C$ and taking the limit as $C\to\infty$, we obtain \eqref{eq:div_multi_pro}.

\section{Proof of Theorem \ref{th:S1}}
\label{app:S1}
We have by the definition of $\outa$ in Scenario 1 that
\begin{equation*}
P(\outa)=P(S^m(n)+Q^u(n)>C).
\end{equation*}
By Cramer's theorem, we have
\begin{equation}
\label{eq:S1_Cramer}
\dl_1(\gu,\gm,\thl)=\inf_{r>0}\{r-\Lambda_{m+u}(r)\},
\end{equation}
where $$\Lambda_{m+u}(r)=\gu(e^r-1)+\thl\log\left(e^{-\frac{\gm}{\thl}}+e^r-e^{r-\frac{\gm}{\thl}}\right).$$

Differentiating $r-\Lambda_{m+u}(r)$ with respect to $r$ and equating with $0$, we obtain
\begin{equation}
\label{eq:S1_eq}
\gu\left( e^{\frac{\gm}{\thl}} -1\right)e^{2r^*}+\left((\thl-1)e^{\frac{\gm}{\thl}}-\thl+\gu+1\right)e^{r^*}-1=0.
\end{equation}
Set $y_1=e^{r^*}$, then \eqref{eq:S1_eq} is a quadratic equation in $y_1$, that can be solve analytically for two possible roots. Choosing the root $y_1>1$ for $r^*>0$, we get
\begin{equation*}
\begin{split}
y_1=&\frac{1}{2\gu \left(e^{\frac{\gm}{\thl}}-1\right)}\Biggl[\Bigl((\thl^2-2\thl+1)e^{\frac{2\gm}{\thl}}\\
   &+\bigl(-2\thl^2+2\thl(\gu+2)+2(\gu-2)\bigr)e^{\frac{\gm}{\thl}}+\thl^2\\
   &-2\thl(\gu+1)+{\gu}^2-2\gu+1\Bigr)^{\frac{1}{2}}+(1-\thl)e^{\frac{\gm}{\thl}}\\
   &+\thl-\gu-1\Biggr].
\end{split}
\end{equation*}
Substitution with $y_1=e^{r*}$ into \eqref{eq:S1_Cramer}, we obtain \eqref{eq:S1}.

\section{Proof of Theorem \ref{th:S2}}
\label{app:S2}

Under the policy $\pi_2$, suppose that an outage event has occurred in slot $n\gg 1$, then $N^m_0(n)+Q^u(n)>C$, which can be decomposed to either of the following to events: 1) $Q^u(n)>C$ or 2) $Q^u(n)\leq C$ but $N^m_0(n)>0$ so that $N^m_0(n)+Q^u(n)>C$. 
Now, focus on the second event, specifically, $N^m_0(n)>0$. To each data source of the $N^m_0(n)$, at least one corresponding request has already arrived at slot $n-T$. Since $N^m_0(n)>0$ and $N^m_0(n)+Q^u(n)>C$, then the system is operating at full capacity in the slots $[n-T,\cdots,n]$. That is, $$N^m(n-T)+\sum_{i=0}^{T}Q^u(n-i)>C(T+1),$$ where $N^m(n-T)$ is the number of distinct multicast data sources demanded by at least one request existing in the system at slot $n-T$.

From \eqref{eq:Aux_ineq}, $N^m(n-T)\leq Z^m_T(n-T)$, where $Z^m_T(n-T)$ is as defined in Appendix \ref{app:div_multi_pred}, then we can now write
\begin{align*}
P_2(\outa)&\leq P(Q^u(n)>C)\\
          &\quad +P\Biggl(\sum_{i=0}^{T}Q^u(n-i)+Z^m_T(n-T)>C(T+1),\\
          &\quad Q^u(n)<C\Biggr)\\
          &\leq P(Q^u(n)>C)\\
          &\quad +P\left(\sum_{i=0}^{T}Q^u(n-i)+Z^m_T(n-T)>C(T+1)\right).
\end{align*} 

We have from Theorem \ref{th:diversity_reac} that \begin{equation}\label{eq:p1}\lim_{C\to\infty}-\frac{\log P(Q^u(n)>C)}{C}=\gu-1-\log \gu.\end{equation} Also, Cramer's theorem can be used in the same way of Appendix \ref{app:S1} to show that

\begin{align}
\label{eq:p2}
& \quad \lim_{C\to\infty}-\frac{1}{C}\log P\Biggl(\sum_{i=0}^{T}Q^u(n-i)+Z^m_T(n-T)\nonumber\\
&\quad \quad >C(T+1)\Biggr)\nonumber\\
&= (T+1)\log y_2-(T+1)\gu(y_2-1) \nonumber\\
&\quad-\thl\log(1-\xm+\xm y_2),
\end{align}

where
\begin{equation*}
\begin{split}
y_2=&\frac{1}{2\xm\gu (T+1)}\Biggl[\Biggl(\Bigl((1-\xm)^2{\gu}^2+2\xm\gu(1-\xm)\\
    & +{\xm}^2\Bigr)^2T^2 +\Bigl([2\xm\gu(1-\xm)-2{\xm}^2]\thl\\
    &+2{\xm}^2(1-\xm)^2+4\xm\gu(1-\xm)+2{\xm}^2\Bigr)T\\
    &+[2\xm\thl(1-\xm)-2{\xm}^2]\thl+{\gu}^2(1-\xm)^2\\
    &+2\xm\thl(1-\xm)+{\xm}^2(1+\thl)^2\Biggr)^{\frac{1}{2}}\\
    &+\Bigl((\xm-1)\gu\Bigr)T-\xm\thl+\gu(\xm-1)+\xm\Biggr].
\end{split}
\end{equation*}
Therefore, from \eqref{eq:p1} and \eqref{eq:p2}, \eqref{eq:S2_l2} follows.

To see  \eqref{eq:S2_u2}, it suffices to note that $Q^u(n)>C$ is a sufficient condition for an outage at slot $n$ independently of the service policy used.
Hence, $P_2(\outa)\geq P(Q^u(n)>C)$, therefore,
$\dl_2(\gu,\gm,\thl)\leq \dl_N(\gu)$.

\section{Proof of Theorem \ref{th:S3}}
\label{app:S3}
An outage event at slot $n$ implies $N^u(n-T)+N^m(n-T)>C(T+1)$ where $N^u(n-T)$ is the number of unicast requests existing in the network at time slot $n-T$. Hence
\begin{equation*}
P_3(\outa)\leq P(N^u(n-T)+N^m(n-T)>C(T+1)),
\end{equation*}
but $$N^u(n-T)\leq \sum_{i=0}^{T}Q^u(n-i-T),$$ and $$N^m(n-T)\leq Z^m_T(n-T).$$
Therefore
\begin{equation*}
P_3(\outa)\leq P\Biggl(\sum_{i=0}^{T}Q^u(n-i-T)+Z^m_T(n-T)>C(T+1)\Biggr).
\end{equation*}

Since $\{Q^u(i)\}_i$ are IID random variables, then from \eqref{eq:p2}, we obtain \eqref{eq:S3_l3}.

\section{Proof of Theorem \ref{th:S4}}
\label{app:S4}
Regardless of the scheduling policy used, the following event is sufficient for an outage at slot $n$.
\begin{equation*}
\begin{split}
& Q^u(n-i-T)>2C-S^m(n-2i)-S^m(n-2i+1),\\
&\hspace{2.5 in} i=1,\cdots,T,\\
&\text{and}\\
& Q^u(n-T)>C-S^m(n).
\end{split}
\end{equation*}
The above event ensures that the number of delayed unicast requests is increasing over the window of slots $[n-2T,\cdots,n-T]$ where at slot $n-T$, the network will end up having $$\sum_{i=0}^T Q^u(n-i-T)+S^m(n-i)>C(T+1),$$ implying that the total number of resources that have to be consumed by slot $n$ inclusive is greater than the aggregate available capacity $C(T+1)$ which would cause an outage.

Noting that $\{S^m(i)\}_i$ are IID, we can write
\begin{equation*}
\begin{split}
P^*_4(\outa)\geq &  P(Q^u(n-T)+S^m(n)>C)\\
                 & \times P\Bigl( Q^u(n-T+1)+S^m(n-2)\\
                 & +S^m(n-1)>2C\Bigr)^T,
\end{split}
\end{equation*}
which, using Chernoff bound, leads to \eqref{eq:S4_u4}.

\end{document}